\documentclass[twocolumn,10pt]{IEEEtran}
\usepackage{braket}
\usepackage{dsfont}
\usepackage{pifont}

\input{def.tex}

\usepackage{dsfont}
\usepackage{minibox,physics}
\DeclareSymbolFont{matha}{OML}{txmi}{m}{it}
\DeclareMathSymbol{\varv}{\mathord}{matha}{118}
\usepackage{eurosym}
\usepackage{multicol}
\usepackage{algorithm,algorithmicx}
\usepackage{eurosym}

\makeatletter

\usepackage{tikz}
\usepackage{array}
\usepackage{booktabs,adjustbox}

\newcommand{\finalcells}[2]{%
	\begingroup\sbox0{\begin{minipage}{3cm}\raggedright#1\end{minipage}}%
	\sbox2{\begin{minipage}{3cm}\raggedright#2\end{minipage}}%
	\xdef\finalheight{\the\dimexpr\ht0+\dp0+\smallskipamount\relax}%
	\xdef\finalheightB{\the\dimexpr\ht2+\dp2+\smallskipamount\relax}%
	\ifdim\finalheightB>\finalheight
	\global\let\finalheight\finalheightB
	\fi\endgroup
	\begin{minipage}[t][\finalheight][t]{3cm}\raggedright#1\end{minipage}&
	\begin{minipage}[t][\finalheight][t]{3cm}\raggedright#2\end{minipage}}

\IEEEoverridecommandlockouts
\begin{document}
	\title{Impact of Channel Aging on Reconfigurable Intelligent
		Surface Aided Massive MIMO Systems with Statistical CSI}
	\author{Anastasios Papazafeiropoulos, Ioannis Krikidis, Pandelis Kourtessis \thanks{A. Papazafeiropoulos is with the Communications and Intelligent Systems Research Group, University of Hertfordshire, Hatfield AL10 9AB, U. K., and with the SnT at the University of Luxembourg, Luxembourg. I. Krikidis is with the IRIDA Research Centre for Communication Technologies, Department of Electrical and Computer Engineering, University of Cyprus, Cyprus. P. Kourtessis is with the Communications and Intelligent Systems Research Group, University of Hertfordshire, Hatfield AL10 9AB, U. K. A. Papazafeiropoulos was supported  by the University of Hertfordshire's 5-year Vice Chancellor's Research Fellowship. Also, this work was co-funded by the European Regional Development Fund and the Republic of Cyprus through the Research and Innovation Foundation under the project INFRASTRUCTURES/1216/0017 (IRIDA). It has also received funding from the European Research Council (ERC) under the European Union’s Horizon 2020 research and innovation programme (Grant agreement No. 819819).
			E-mails: tapapazaf@gmail.com, krikidis.ioannis@ucy.ac.cy, p.kourtessis@herts.ac.uk.}}
	\maketitle\vspace{-1.7cm}
	\maketitle

	\begin{abstract}	
		The incorporation of reconfigurable intelligent surface (RIS) into massive multiple-input-multiple-output (mMIMO) systems can unleash the potential of next-generation networks by improving the performance of  user equipments (UEs) in service dead zones. However, their requirement for accurate channel state information (CSI) is critical, and especially,  applications with UE mobility  that induce channel aging make challenging the achievement of adequate quality of service. Hence, in this work, we investigate the impact of channel aging on the performance of RIS-assisted mMIMO systems under both spatial correlation and imperfect CSI conditions. Specifically, by accounting for channel aging during both uplink training and downlink data transmission phases, we first perform  minimum mean square error (MMSE) channel estimation to obtain the UE effective channels with low overhead similar to conventional systems without RIS. Next, we derive the downlink achievable sum spectral efficiency (SE) with  regularized zero-forcing (RZF) precoding in closed-form  being dependent only on large-scale statistics by using the deterministic equivalent (DE) analysis. Subsequently, we present the attractive optimization of the achievable sum SE with respect to the phase shifts and the total transmit power that can be performed every several coherence intervals due to the slow variation of the large-scale statistics. Numerical results validate the analytical expressions and demonstrate the performance while allowing the extraction of insightful design conclusions for common scenarios including UE mobility. In particular,  channel aging degrades the performance but its impact can be controlled by choosing appropriately the frame duration or by increasing the number of RIS elements.	\end{abstract}
	
	\section{Introduction}
	Emerging applications bring challenges that demand ever-higher data rates and increased connectivity/coverage together with ultra-reliable and low-latency wireless communication (URLLC) requirements. Especially, these applications have led to the development of  disruptive technologies such as massive multiple-input multiple-output (mMIMO) systems and millimeter-wave (mmWave) communications \cite{Boccardi2014}.  Unfortunately, existing techniques incur additional power and hardware costs while they cannot guarantee an adequate quality of service (QoS) in dead zones due to obstacles. For example, mMIMO exhibits poor performance in low scattering conditions, and  the large number of active elements might result in prohibitive energy usage. In particular, they focus on improvements regarding the transmission and reception, while the wireless propagation environment is left uncontrollable.   Furthermore, the time-varying and random nature of the wireless channel constitute the ultimate impediment to achieving the URLLC and rate targets.   
	
	In this direction, sixth generation (6G) networks, aiming at covering the higher rate demands and more stringent constraints, have appeared with the reconfigurable intelligent surface (RIS) being among its proposed promising technologies. Actually, RIS   has attracted significant attention since it overcomes the aforementioned issues \cite{Basar2019,Wu2019,Pan2020,Papazafeiropoulos2021,Kammoun2020, Papazafeiropoulos2021b,Elbir2020,Guo2020,Chen2019,Yang2021,Bjoernson2019b,DiRenzo2020}. Specifically, a RIS is a software-defined surface that is usually attached  to existing infrastructure to alleviate blockage effects. It consists of a large number of  individually-controlled, low-cost, and nearly passive elements. A RIS achieves to adapt to changes in the propagation environment and modify the radio waves since each of its elements can induce an adjustable phase shift to each incident signal, which enables a dynamic control over the wireless propagation channel. For instance, in \cite{Wu2019}, a minimization of the transmit power at the base station (BS) with signal-to-interference-plus-noise ratio (SINR) constraints took place in a RIS-assisted multi-user (MU) multiple-input single-output (MISO) communication system by  jointly optimizing the precoding and reflecting beamforming matrices (RBMs). In \cite{Pan2020}, the sum rate was maximized subject to a transmit power constraint, while in \cite{Papazafeiropoulos2021}, the sum rate was optimized by accounting also for correlated Rayleigh fading and inevitable hardware impairments at both the transceiver and the RIS. Similarly, in \cite{Kammoun2020}, the maximization of the minimum UE rate was studied in the case of a large number of antennas,  in \cite{Papazafeiropoulos2021b}, the impact of hardware impairments was evaluated, and in \cite{Yang2021}, the impact of imperfect CSI on the outage probability was investigated. Notably, given that RIS-assisted MIMO systems, having a reduced number of active  radio frequency (RF) chains,  can achieve similar performance to mMIMO without RIS \cite{DiRenzo2020}, it is indicated that a more cost and energy-efficient implementation of mMIMO is possible. 
	
	To reap the benefits of RIS and mMIMO and arrive at realistic conclusions, the acquisition of accurate channel state information (CSI) is of paramount importance \cite{Mishra2019,He2019,Elbir2020,Nadeem2020,Zheng2019,Shtaiwi2021}. \footnote{Note that many previous works assumed perfect CSI, which is a highly unrealistic assumption.} However, channel estimation (CE) in RIS-aided systems is quite challenging because of two main reasons. First, although its passive elements render RIS energy-efficient, they make infeasible conventional CE through transmitting and receiving pilots, which require active elements. Second, RIS generally includes a large number of elements, which require a prohibitively high training overhead that severely reduces  the achievable rate. For instance, in  \cite{Mishra2019}, an ON/OFF CE scheme was proposed, where   the least-squares  estimates of all RIS-assisted MISO channels   with  a single user were calculated one by one. Moreover, other works  such as \cite{He2019, Elbir2020} do not provide analytical expressions for the estimated channel that could be exploited for the derivation of the spectral efficiency (SE).  In \cite{Nadeem2020}, all RIS elements were assumed active during training but a number of sub-phases equal at least to the number of RIS elements are required, which results in a lower rate because the overhead on the coherence time for CE is larger. Moreover, this method does not provide the covariance of the estimated channel vector from all RIS elements to a specific UE but estimates of the individual channels while leaving the correlation among them unknown. An  effective method with low overhead compared to previous works is to estimate the cascaded BS-RIS-UE channel  in a single phase based on  minimum mean square error (MMSE) as in \cite{Papazafeiropoulos2021}.  Notably, another method, which reduces the overhead has been presented by exploiting RIS partitioning into subgroups, e.g., see \cite{Shtaiwi2021}. Also, therein, an insightful categorization of the various CE approaches concerning RIS-assisted systems has been provided.

	In practice, CSI is not only imperfect but can also be outdated because of channel aging \cite{Truong2013,Papazafeiropoulos2015a,Papazafeiropoulos2016,Chopra2017,Chen2021}. The cause of channel aging is the UE mobility, which renders the channel time-varying, i.e., contrary to the standard block fading model, the channel evolves with time and is different during each symbol. Thus, a mismatch appears between the current channel and the estimated channel used for detection or precoding. Interestingly, in  \cite{Chopra2017}, channel aging was also considered during the training phase for a more realistic study. Several works have studied the impact of channel aging in mMIMO systems as mentioned but little attention has been given to its effect in  RIS-assisted systems despite its great significance \cite{Chen2021}.
	
	In principle, the phase shifts  optimization lies on two methodologies with respect to CSI, namely, instantaneous CSI (I-CSI) \cite{Wu2019,Pan2020} and statistical CSI (S-CSI) \cite{Zhao2020,Kammoun2020,Papazafeiropoulos2021,VanChien2021,Papazafeiropoulos2021a,Chen2021,You2021,Zhi2022,Zhang2022}. The first approach suggests the optimization of the phases at every coherence interval because the related expressions depend on small-scale fading, while the second approach concerns expressions that depend on  large-scale statistics, which vary every several coherence intervals. Hence, the latter approach enables considerably the reduction of the signal overhead and the computational complexity, which can become excessively high in the case of a large number of RIS elements and BS antennas. Moreover, for the same reasons, the S-CSI approach is more energy-efficient. Notably, in high mobility scenarios, which are faster time-varying, the I-CSI method would be very challenging to be implemented since the tuning of the RIS parameters should be repeated very frequently. On the contrary, the application of the  S-CSI appears to be more practical.
	
	\subsection{Motivation/Contributions}
	Faced with these challenges, the motivation of this work is to conduct a realistic characterization of the downlink achievable sum SE of RIS-assisted mMIMO systems accounting for UE mobility  and imperfect CSI under correlated Rayleigh fading conditions, when  regularized zero-forcing (RZF) precoding is applied.
	\begin{itemize}
		\item Contrary to the majority of existing works on RIS-assisted systems such as \cite{Wu2019,Pan2020,Zhao2020,Kammoun2020,Zheng2019,Shtaiwi2021,Papazafeiropoulos2021,VanChien2021,Papazafeiropoulos2021a,You2021,Zhi2022}, which  considered static UEs, we account for channel aging due to UE mobility. To the best of our knowledge, the only  previous works considering channel aging are \cite{Chen2021} and \cite{Zhang2022}. The former focused on  mmWave communications with LoS links and with a finite number of BS antennas, while we consider mMIMO systems, correlated Rayleigh fading, and channel aging during the training phase too.  The latter did not account for correlated fading  mMIMO, and ZF precoding, while we have taken correlation into account and we have focused on mMIMO. Note that no optimization took place in \cite{Zhang2022}. In addition, we have resorted to the deterministic equivalent analysis to provide the rate for RZF. Also, despite that many previous works have relied on independent Rayleigh fading e.g., \cite{Wu2019,Pan2020}, we consider correlated Rayleigh fading \cite{Bjoernson2020}, which appears unavoidable in practice. Moreover, we account for S-CSI instead of I-CSI since the former is more suitable for studying time-varying channels. In particular,  compared to other works, which are based on statistical CSI such as \cite{Zhao2020,Kammoun2020,Papazafeiropoulos2021,VanChien2021,Papazafeiropoulos2021a,Chen2021,You2021,Zhi2022,Zhang2022}, our work is the only one that has studied the impact of channel aging by taking into account correlated fading, imperfect CSI, and RZF being a more advanced precoder, which increases the difficulty for the derivation of closed-form expressions. For example, compared to \cite{Papazafeiropoulos2021} focusing on the uplink and maximal ratio combining (MRC), we have assumed a more suitable model for RIS correlation, the downlink, RZF, and we have focused on the impact of channel aging. Similarly, compared to \cite{Kammoun2020}, we have assumed imperfect CSI, have focused on the sum rate instead of the max-min rate, and have studied channel aging. Also, although channel aging has been studied in \cite{Zhang2022}, no correlation has been considered, ZF instead of RZF has been applied, and no optimization has been performed.
		\item We introduce channel aging not only in the downlink data transmission phase but also during the uplink training phase as in \cite{Chopra2017}. In particular, based on \cite{Papazafeiropoulos2021}, we perform MMSE  estimation and obtain the effective channel estimate that ages with time. The proposed approach provides the estimated channel with low overhead and in closed-form that enables further manipulations to derive the achievable SE. Previous works, e.g., \cite{He2019, Elbir2020} do not provide analytical expressions or have other disadvantages such as high overhead and unknown inter-element correlation \cite{Nadeem2020}.
		\item Exploiting the deterministic equivalent (DE) analysis, we obtain the DE of the downlink sum SE of RIS-assisted mMIMO systems with RZF precoding under UE mobility and correlated Rayleigh fading conditions.  The DE results are of great importance because they provide closed-form expressions in terms of a convergent system of fixed-point 	equations that allow efficient optimization. 
		\item We formulate the maximization problem regarding the sum SE with respect to RBM and total transmit power constraints. Notably, given that the sum SE depends only on large-scale statistics, the proposed optimization  can be performed every several coherence intervals, and thus, reduce significantly the signal overhead, which is large in time-varying channels.
		\item We verify the analytical results with Monte Carlo (MC) simulations, and we  shed light on the impact of channel aging on the downlink sum SE of a RIS-assisted mMIMO system  due to correlated fading and channel aging. For comparison, we depict results corresponding to no mobility to show the degradation due to channel aging and the inferior performance of maximum ratio transmission (MRT) precoding.
	\end{itemize}

	\subsection{Paper Outline} 
	The remainder of this paper is organized as follows. Section~\ref{SystemModel} presents the system model of a RIS-assisted mMIMO system  with imperfect CSI under correlated Rayleigh fading and  channel aging conditions. Section~\ref{ChannelEstimation} describes the CE accounting for channel aging. Section~\ref{PerformanceAnalysis} presents the downlink sum SE, while Section~\ref{SumSEMaximizationDesign} provides the optimization regarding the  RBM and the transmit power.
	The numerical results are discussed in Section~\ref{Numerical}, and Section~\ref{Conclusion} concludes the paper.
	\begin{figure}[!h]
		\begin{center}
			\includegraphics[width=0.9\linewidth]{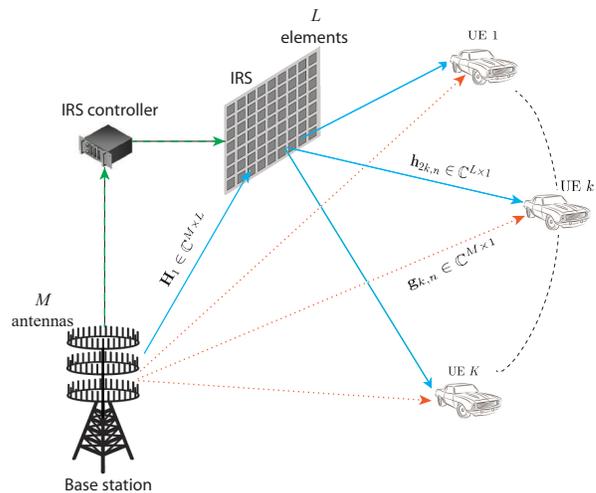}
			\caption{\footnotesize{ A downlink RIS-assisted mMIMO communication system with $ M $ BS antennas, $ L $ RIS elements, and moving $ K $ UEs.  }}
			\label{Fig0}
		\end{center}
	\end{figure}
	\subsection{Notation}Vectors and matrices are denoted by boldface lower and upper case symbols, respectively. The notations $(\cdot)^\T$, $(\cdot)^\H$, and $\tr\!\left( {\cdot} \right)$ represent the transpose, Hermitian transpose, and trace operators, respectively. The expectation operator is denoted by $\EE\left[\cdot\right]$  while $ \diag\left(\ba \right) $ represents an $ n\times n $ diagonal matrix with diagonal elements being the elements of vector $ \ba $. 
	Also, the notations $\xrightarrow[ M \rightarrow \infty]{\mbox{a.s.}}$ and $a_n\asymp b_n$ with $a_n$ and $b_n$ being two infinite sequences denote almost sure convergence as $ M \rightarrow \infty $. The notation $ \displaystyle\lim_{x\to c}  f(x) $ denotes the limit of $ f $ of $ x $ as $ x $ approaches $ c $, and the notation $  \pdv{f(x)}{x}   $ denotes the partial derivative of $ f $ with respect to  $ x $.  Finally, $\bb \sim \cC\cN{(\b0,\mathbf{\Sigma})}$ represents a circularly symmetric complex Gaussian vector with {zero mean} and covariance matrix $\mathbf{\Sigma}$.
	\section{System Model}\label{SystemModel}
	We consider a RIS-assisted downlink mMIMO  system, where a BS, equipped with $ M $ antennas communicates with $ K $  single-antenna noncooperative UEs behind obstacles. A RIS, consisting of $ L $ passive reflecting elements is located in the LoS of the BS to assist the communication with the UEs, e.g., imagine the common scenario where   both the BS and RIS are deployed at high altitude and their locations are fixed, as shown in Fig. \ref{Fig0}. The RIS can  dynamically adjust the phase shift induced by each reflecting element on the impinging electromagnetic waves through a perfect  smart controller that is connected to  the BS in terms of a perfect backhaul link. The size of each RIS element is $ d_{\mathrm{H}}\times d_{\mathrm{V}} $, where $d_\mathrm{V}$ and $d_\mathrm{H}$ express its vertical height and its horizontal width, respectively. The proposed model considers also the potential presence of direct links between the BS and the UEs. However, these could be also neglected in the cases of high penetration losses and/or big signal blockages.
	
	\subsection{Channel Model}\label{ChannelModel} 
	We account for a quasi-static fading model with coherence bandwidth much larger than the channel bandwidth. We employ the standard block fading model with each coherence interval/block including $\tau_{\mathrm{c}}= B_{\mathrm{c}}T_{\mathrm{c}}$ channel uses, where $ B_{\mathrm{c}} $ and $ T_{\mathrm{c}} $ are the coherence bandwidth and the coherence time in $\mathrm{Hz}$ and $\mathrm{s}$, respectively.

	Within the transmission in each coherence block and during the $ n $th time slot, let  $ \bH_{1}=[\bh_{11}\ldots,\bh_{1L} ] \in \mathbb{C}^{M \times L}$ and $ \bh_{2k,n} \in \mathbb{C}^{L \times 1}$ be  the LoS channel between the BS and the RIS and the channel between the RIS and UE $ k $ at the $ n $th time instant. Note that  $ \bh_{1i} $ for $ i=1,\ldots,L $ denotes the $ i $th column vector of $ \bH_{1} $.	Similarly, let $ \bg_{k,n} \in \mathbb{C}^{M \times 1} $ be the direct channel between the BS and UE $ k $ at the $ n $th time instant.  Despite that  the majority of existing works, e.g., \cite{Wu2019,Pan2020}, assumed independent Rayleigh model, in practice, correlated fading appears, which affects the performance  \cite{Bjoernson2020}.		Thus, $ \bh_{2k,n} $ and  $ \bg_{k,n} $ are described in terms of correlated Rayleigh fading distributions as
	\begin{align}
		\bh_{2k,n}&=\sqrt{\beta_{\mathrm{h}_{2},k}}\bR_{\mathrm{RIS},k}^{1/2}\bq_{2k,n},\\
		\bg_{k,n}&=\sqrt{\beta_{\mathrm{g},k}}\bR_{\mathrm{BS},k}^{1/2}\bq_{k,n},
	\end{align}
	where $ \bR_{\mathrm{RIS},k} \in \mathbb{C}^{L \times L} $ and $ \bR_{\mathrm{BS},k} \in \mathbb{C}^{M \times M} $ express the deterministic Hermitian-symmetric positive semi-definite correlation
	matrices at the RIS and the BS respectively with $ \tr\left(\bR_{\mathrm{RIS},k} \right)=L $ and $ \tr\left(\bR_{\mathrm{BS},k} \right)=M $.  The correlation matrices $ \bR_{\mathrm{RIS},k} $ and $ \bR_{\mathrm{BS},k}~\forall k$ are assumed to be known by the network since they can be obtained by  existing estimation methods \cite{Neumann2018}.\footnote{The correlation matrices and the path-losses are independent of  $n$ because these represent effects that vary with time in a much slower pace than the coherence time.} Another way of practical calculation of
		the covariance matrices follows. Especially, as can be seen by the expression of the covariance matrices, they depend on
		the distances $ d_{\mathrm{BS}} $ and $ d_{\mathrm{IRS}} $ and the angles. The distances are based on the construction of the BS and the IRS.
		Moreover, the angles can be calculated when the locations are given. Moreover,  $ \beta_{\mathrm{h}_{2},k} $ and $ \beta_{g,k} $ describe the path-losses  of the RIS-UE $ k $ and BS-UE $ k $ links, respectively. Especially, $ \beta_{g,k} $ is expected to be small because of the blockages between the BS and the UEs. Also, $ \bq_{2k,n}\sim \mathcal{CN}\left(\b0,\Id_{L}\right) $ and $ \bq_{k,n} \sim \mathcal{CN}\left(\b0,\Id_{M}\right) $ denote the corresponding fast-fading vectors at the $ n $th time instant. Note that fast fading vectors change within each coherence block, while the correlation matrices are assumed constant for a large number of coherence blocks.

	The high rank LoS channel $ \bH_{1} $ is described as
	\begin{align}
		[\bH_{1}]_{m,l}=\sqrt{\beta_{1}} \exp\big(j \frac{2 \pi }{\lambda}\left(m-1\right)d_{\mathrm{BS}}\sin \theta_{1,l}\sin \phi_{1,l}\nn\\
		+\left(l-1\right)d_{\mathrm{RIS}}\sin \theta_{2,m}\sin \phi_{2,m}\big)\!,
	\end{align}
	where  $ \beta_{1} $ is the path-loss between the BS and RIS, $ \lambda $ is the carrier wavelength, while $ d_{\mathrm{BS}} $ and $ d_{\mathrm{RIS}} $ are the inter-antenna separation at the BS and inter-element separation at the RIS, respectively \cite{Nadeem2020}. Also, $ \theta_{1,l} $ and $ \phi_{1,l} $ denote the elevation and azimuth LoS angles of departure (AoD) at the BS with respect to RIS element $ l $, and $ \theta_{2,l} $ and $ \phi_{2,l} $ denote the elevation and azimuth LoS angles of arrival (AoA) at the RIS. It is worthwhile to mention that  $ \bH_{1} $ can be obtained similarly to the covariance matrices
		since the dependence of their expressions on the distances and the angles is similar.

	The response of the $ L $ elements is described by the diagonal RBM    $\bTheta=\mathrm{diag}\left( \mu_{1}e^{j \theta_{1}}, \ldots, \mu_{L}e^{j \theta_{L}} \right)\in\mathbb{C}^{L\times L}$, where $ \theta_{l} \in [0,2\pi]$ and $ \mu_{l}\in [0,1] $ are the phase  and amplitude coefficient for RIS element $ l $, respectively. Herein, we assume maximum signal reflection, i.e., $ \mu_{l}=1~ \forall l$ \cite{Wu2019}.\footnote{Recently, it was shown that the amplitude and phase responses are 	intertwined in practice, while the assumption of  independence between the amplitude and the phase shift or even a unity amplitude is unrealistic, e.g., see \cite{Abeywickrama2020}. However, this 	assumption regarding independence still allows revealing fundamental properties of the channel aging of the proposed model, while the consideration of the phase shift model in \cite{Abeywickrama2020} is an interesting idea for extension of the current work, i.e., to study the impact of  channel aging on RIS-assisted systems by accounting for this intertwinement.    } For the sake of exposition, the overall channel vector $ \bh_{k,n}=\bg_{k,n}+ \bH_{1}\bTheta \bh_{2k,n} $, conditioned on $ \bTheta $ is distributed as $ \bh_{k,n}\sim \cC\cN\left( 0, \bR_{k} \right) $, where $ \bR_{k}= \beta_{\mathrm{g},k}\bR_{\mathrm{BS},k}+ \beta_{\mathrm{h}_{2},k}\bH_{1} \bTheta {\bR}_{\mathrm{RIS},k}\bTheta^{\H}\bH_{1}^{\H}$. Given that ${\bf R}_k$ depends on the path-losses, the correlation matrices, and $ \bH_{1} $, which are all assumed to known as explained previously, ${\bf R}_k$ can also be assumed known by the network.

	\begin{remark}\label{rem1}
		Although, it is uncommon to meet  independent Rayleigh fading in practice \cite{Bjoernson2020}, in such a case, we have $ {\bR}_{\mathrm{RIS},k}=\Id_{L}$. Then, the overall covariance becomes $ \bR_{k} =\beta_{\mathrm{d},k}\bR_{\mathrm{BS},k}+\beta_{2,k}\bH_{1}\bH_{1}^{\H}$. Obviously,  $ \bR_{k} $ does not depend on the RBM and cannot be optimized. Hence, the RIS cannot be exploited. 	Nevertheless, note that even under these conditions, the  RIS enhances  the communication with an additional signal to the receiver.
	\end{remark}
	
	\subsection{Channel Aging}
	In practice, the relative movement between the UEs and the RIS, i.e., the UE mobility causes a phenomenon, known as channel aging~\cite{Truong2013,Papazafeiropoulos2015a,Papazafeiropoulos2016}.\footnote{Normally, all RIS elements have the same relative movement comparing to a specific UE.} In particular, this movement results in a Doppler shift that makes the channel change with time. Hence, contrary to the conventional block fading channel model, the channel coefficients, exhibiting flat fading, vary from symbol to symbol. However, they are constant within one symbol.  The symbol duration is assumed smaller than or equal to the coherence time of all UEs. This assumption is common in works studying the impact of channel aging such as \cite{Truong2013}-\cite{Chopra2017}. The channel use is denoted by $ n \in \{ 1, \ldots, \tau_{\mathrm{c}}\} $.

	Mathematically,  the channel realization $ \bh_{k,n} $ at the $ n $th time instant is modeled as a function of its
	initial state $ \bh_{k,0} $ and an innovation component as \cite{Chopra2017}
	\begin{align}
		\bh_{k,n}=\al_{k,n}\bh_{k,0}+\bar{\al}_{k,n}\bee_{k,n},\label{GaussMarkoModel}
	\end{align}
	where $ \bee_{k,n}\sim\mathcal{CN}\left(\b0,\bR_{k}\right)$  denotes  the independent innovation component at the $ n $th time instant and $ \al_{k,n}= \mathrm{J}_{0}(2 \pi f_{\mathrm{D}}T_{\mathrm{s}}n) $ is the temporal correlation coefficient of UE $ k $  between the channel realizations at time $ 0 $ and $ n $ with  $\mathrm{J}_{0}(\cdot)$ being the zeroth-order Bessel function of the first kind, $T_{s}$ being the channel sampling duration,  and $f_{D}=\frac{vf_{c}}{c}$ being the maximum Doppler shift.\footnote{The second-order statistics of the channel including the path-losses and $\al_{k,n}$ are estimated during the connection establishment when the BS estimates the location and the velocity of the UE. While the path-losses are estimated using the average pilot power, $ \al_{k,n} $ can be estimated using the temporal correlation among the same set of pilots.} Also, we denote $ \bar{\al}_{k,n}=\sqrt{1-\al^{2}_{k,n}}$. Note that $v$ is the velocity of the UE, $c = 3 \times 10^{8}$ $\nicefrac{m}{s}$ is the speed of light, and $f_{c}$ is the carrier frequency. As can be seen, a higher UE velocity of the UE  or higher delay result in decrease of $ \al_{k,n} $ though not monotonically, since there are some ripples. It is worthwhile to mention that the model in \eqref{GaussMarkoModel}  is not  autoregressive of first-order as in previous works \cite{Truong2013,Papazafeiropoulos2015a} since the current channel is not determined in terms of its state   at the previous time instant,  but it depends on its state at an initial time $ \bh_{k,0} $. The advantage is that the statistics of the model exactly match with that of the Jakes’ model \cite{Chopra2017}. 
	
	\section{Channel Estimation with Channel Aging}\label{ChannelEstimation}
	Perfect CSI is not available in practice but the BS needs to estimate the channel. On this ground,  we consider the standard time-division-duplex (TDD) protocol, where each block consists of $\tau$ channel uses for the uplink training phase and $\tau_{\mathrm{d}}=\tau_{\mathrm{c}}-\tau$ channel uses for the downlink data transmission phase~\cite{massivemimobook}. The disadvantage of RIS-assisted systems is that the RIS, which consists of passive elements, cannot  process the received pilot symbols from the UEs to obtain the  estimated channels and cannot  send pilots to the BS for CE. Herein, contrary to ON/OFF channel estimation schemes such as  \cite{Mishra2019} and \cite{Nadeem2020} that require $ L+1 $ phases,  we perform the CE in a single phase.  Notably, the consideration of  the cascaded channel CE instead of the individual channels has already been applied in several works such as \cite{Papazafeiropoulos2021,Kammoun2020}. Actually, it is more beneficial to consider the overall channel because this allows computing the correlation among inter-element links, while in the case of individual channels, this correlation remains unknown. Also, our CE is accompanied by reduced feedback and allows higher achievable SE  due to the larger pre-log factor since the training overhead is much lower.
	
	During the uplink training phase, 	each UE transmits a $\tau$-length
	mutually orthogonal training sequences, i.e., $\bm \psi_{k}= \left[\psi_{k,1},\ldots,\psi_{k,\tau} \right]^{\T}\in \bbC^{\tau \times 1}$. We assume that the pilot sequence consists of $ K $  pilot symbols, which is the  minimal number for channel estimation, i.e., $ \tau \ge K $ \cite{Hassibi2003}. Also, $ \bTheta $ is assumed fixed. Thus, the received signal by the BS at time $ n $ is  given by
	\begin{align}
		\bY_{n}=\sqrt{p_{\mathrm{p}}}\sum_{i=1}^{K}\bh_{i,n}\bpsi_{i}^{\H}+\bW_{n},\label{TrainingReceived}
	\end{align}
	where $ p_{\mathrm{p}}\ge 0 $ is the common pilot transmit power for all UEs   and $ \bW_{n}\in \mathbb{C}^{M\times \tau} $ is spatially white additive Gaussian noise matrix at the BS   during this phase. After correlation of the received signal with the training sequence of UE $ k $ $\frac{1}{\sqrt{p_{\rp}}}\bpsi_k$, we obtain
	\begin{align}
		\tilde{\by}_{k,n}=\bh_{k,n}+\frac{1}{\sqrt{p_{\rp}}}\tilde{\bw}_{k,n},\label{TrainingReceived1}
	\end{align}
	where $ \tilde{\bw}_{k,n}=\bW_{n}\bpsi_{k}\sim \mathcal{CN}\left(\b0,\tilde{\sigma}^{2}\Id_{M}\right) $. 
	
	Although this received signal can be used to estimate the channel at any time slot of the block, the estimated channel  will deteriorate as the time interval between training and transmission increases. On this ground, we consider the channel estimate at $ n=K+1 $ since the estimate will be worse at a later instant. Based on \eqref{GaussMarkoModel}, the channel at the $ n $th instant ($ n\le K $) can be described in terms of the channel at time $ K+1 $ as
	\begin{align}
		\bh_{k,n}=\al_{k,\zeta-n}\bh_{k,\zeta}+\bar{\al}_{k,\zeta-n}\tilde{\bee}_{k,n},\label{innov1}
	\end{align}
	where $  \tilde{\bee}_{k,n}\sim\mathcal{CN}\left(\b0,\bR_{k}\right) $ is the  independent innovation vector, which relates $ 	\bh_{k,n} $ and $ \bh_{k,\zeta} $. Also, we have defined $ \zeta=K+1 $ to simplify the notation. Inserting \eqref{innov1} into \eqref{TrainingReceived1}, we obtain
	\begin{align}
		\tilde{\by}_{k,n}&=\al_{k,\zeta-n}\bh_{k,\zeta}+\bar{\al}_{k,\zeta-n}\tilde{\bee}_{k,n}+\frac{1}{\sqrt{p_{\rp}}}\tilde{\bw}_{k,n}\label{TrainingReceived}.
	\end{align}
	
	By applying  the standard minimum mean square error (MMSE) estimation \cite{massivemimobook}, the BS obtains the channel estimate of $ {\bh}_{k,\zeta} $ as
	\begin{align}
		\hat{\bh}_{k,\zeta}=\al_{k,\zeta-n}\bR_{k} \bQ			\tilde{\by}_{k,n},\label{estimatedChannel}
	\end{align}
	where $ \bQ=\left(\bR_{k}+\frac{\tilde{\sigma}^{2}}{{p_{\rp}}}\Id_{M}\right)^{-1} $. The estimate  	$ \hat{\bh}_{k,\zeta} $  is distributed as $ \mathcal{CN}\left(\b0,\bPhi_{k}\right) $, where $ \bPhi_{k}=\al^{2}_{k,\zeta-n}\bR_{k} \bQ \bR_{k}$.  According to the orthogonality property of MMSE estimation, the independent channel estimation error vector is $	\tilde{\bh}_{k,\zeta}=	{\bh}_{k,\zeta}-	\hat{\bh}_{k,\zeta}  $ and distributed as $ \mathcal{CN}\left(\b0,\bPsi_{k}\right) $, where $ \bPsi_{k}=\bR_{k}-\bPhi_{k} $. Notably, the channel estimate 	$ \hat{\bh}_{k,\zeta} $  in \eqref{estimatedChannel} includes the  degradation due to the  channel aging. 
	\begin{remark}
		In the case of no channel aging, i.e., when $ \al_{k,\zeta-n}=1 $, we reduce to the conventional block-fading model. Moreover, as can be seen, the estimation error takes values between $\bR_{k}(\Id_{M}- \al^{2}_{k,1} \bQ \bR_{k})  $ and  $\bR_{k}(\Id_{M}- \al^{2}_{k,K} \bQ \bR_{k})  $. Also, we observe that in the case of no mobility, the estimation error vanishes as the pilot signal-to-noise ratio (SNR) $ \gamma=\frac{p_{\rp}}{\tilde{\sigma}^{2}} $ increases, while it saturates as $\displaystyle\lim_{\gamma\to \infty}\bPsi_{k}= (1-\al^{2}_{k,\zeta-n})\bR_{k} $ in the case of channel aging. The latter shows that the estimation error increases as the UE moves with higher velocity and as the number of UEs increases since $ \zeta  $ increases. 
	\end{remark}
	
	In Sec. \ref{Numerical}, we illustrate the normalized mean square error (NMSE) defined as
	\begin{align}
		\mathrm{NMSE}_{k}&=\frac{\tr\big(\EE[(\hat{\bh}_{k}-{\bh}_{k})(\hat{\bh}_{k}-{\bh}_{k})^{\H}]\big)}{\tr\left(\EE[{\bh}_{k}{\bh}_{k}^{\H}]\right)}\\
		&=1-\frac{\tr(\bPsi_{k})}{\tr(\bR_{k})}.\label{nmse1}
	\end{align}
	
	According to \eqref{nmse1},  an increase in channel aging results in the increase of the $ \mathrm{NMSE}_{k} $. Overall, channel aging has a detrimental on channel estimation. Below, we elaborate on its impact during the downlink transmission.
	
	\section{Downlink Transmission}\label{PerformanceAnalysis}
	The downlink transmission of data from the BS to all UEs consists of a broadcast channel that has to make use of a certain precoding strategy in terms of a precoding vector $\bff_{k,n} \in \bbC^{M \times 1}$. In parallel, taking advantage of TDD and its channel reciprocity, the downlink channel is the Hermitian transpose of the uplink channel. Thus, the received signal $r_{k,n}\in\bbC$ by UE $ k $ during the data transmission phase ($n=K+1,\ldots, \tau_{\mathrm{c}}$)  can be written as
	\begin{align}
		r_{k,n}=\bh^\H_{k,n}\bs_{n}+z_{k,n},\label{DLreceivedSignal}
	\end{align}
	where  $\bs_{n}=\sum_{i=1}^{K}\sqrt{p_{i}}\bff_{i,n}x_{i,n}$ describes the transmit signal vector  by the BS,  $p_{i}\ge 0$ is the transmit power to UE $ i $, and $z_{k,n} \sim \cC\cN(0,\sigma^{2})$ is complex Gaussian noise at UE $k$. Note that   $\bff_{i,n} \in \bbC^{M \times K}$ and 
	$ x_{i,n} $ are  the linear precoding vector  and the data symbol with $ \EE\{|x_{i,n}|^{2}\}=1 $, respectively.\footnote{If we assume that mmWwave communication takes place, hybrid beamforming can be introduced as the best solution that achieves a good trade-off between cost and complexity as usually adopted in the literature. However, the study of the impact of channel aging in the mmWwave  region is left for future research   due to limited space.} The precoding vector is normalized based on the average total power constraint
	\begin{align}
		\EE\{\|\bs_{n}\|^{2}\}=\tr(\bP\bF_{n}^{\H}\bF_{n})\le  P_{\mathrm{max}}\label{PowerConstraint},
	\end{align} 
	where $ \bF_{n}=[\bff_{1,n}, \ldots, \bff_{K,n}] \in\mathbb{C}^{M \times K}$, $ \bP=\diag(p_{1}, \ldots, p_{K})$, and $  P_{\mathrm{max}}>0$ is the total transmit power. However, the channel $ \bh_{k,n} $ in \eqref{DLreceivedSignal} can be expressed as
	\begin{align}
		\bh_{k,n}=&\al_{k,n-K}\hat{\bh}_{k}+\al_{k,n-K}\tilde{\bh}_{k}+\bar{\al}_{k,n-K}\tilde{\bee}_{k,n}\nn\\
		&\al_{k,n-K}\hat{\bh}_{k}+\bar{\al}_{k,n-K}\tilde{\bee}_{k,n}\label{innov2},
	\end{align}
	where $ \hat{\bh}_{k} $ expresses the channel vector at the beginning of the data transmission phase, and $ \tilde{\bh}_{k} $ is the corresponding channel estimation error.
	
	Substitution of \eqref{innov2} into  \eqref{DLreceivedSignal} provides 
	\begin{align}
		r_{k,n}&=\al_{k,n-K}\sqrt{ p_{k}}{\bh}_{k}^{\H}\bff_{k,n}x_{k,n}
		+\bar{\al}_{k,n-K}\sqrt{ p_{k}}\tilde{\bee}_{k,n}^{\H}\bff_{k,n}x_{k,n}\nn\\
		&		+\sum_{i\ne k}^{K}\sqrt{ p_{i}}\bh^\H_{k,n}\bff_{i,n}x_{i,n}+z_{k,n}.\label{DLreceivedSignal1}
	\end{align}
	
	Although  UEs do not have instantaneous CSI, we can assume that UE $ k $ has access to $ \EE\{\hat{\bh}_{k}^{\H}\bff_{k,n}x_{k,n}\} $. Then, by using the technique in~\cite{Medard2000}, where UE $ k $ is aware of only the statistical CSI, the received signal is written as
	\begin{align}
		&r_{k,n}\!=\!\al_{k,n-K}\sqrt{p_{k}}\EE\{{\bh}_{k}^{\H}\bff_{k,n}\}x_{k,n}	\!+\!\!\bar{\al}_{k,n-K}\sqrt{p_{k}}\tilde{\bee}_{k,n}^{\H}\bff_{k,n}x_{k,n}\nn\\
		&+\!\al_{k,n-K}\sqrt{p_{k}}{\bh}_{k}^{\H}\bff_{k,n}x_{k,n}-\al_{k,n-K}\sqrt{p_{k}}\EE\{{\bh}_{k}^{\H}\bff_{k,n}\}x_{k,n}\nn\\
		&+\sum_{i\ne k}^{K}\sqrt{p_{i}}\bh^\H_{k,n}\bff_{i,n}x_{i,n}+z_{k,n}\label{DLreceivedSignal1}.
	\end{align}

	\begin{proposition}\label{LowerBound1}
		The downlink average SE for UE $ k$ of an RIS-assisted mMIMO system, accounting for imperfect CSI and channel aging,  is lower bounded by
		\begin{align}
			\mathrm{SE}_{k}	=\frac{1}{\tau_{\mathrm{c}}}\sum_{n=K+1}^{\tau_{\mathrm{c}}}\log_{2}\left ( 1+\gamma_{k,n}\right)\!,\label{LowerBound}
		\end{align}
		where $ \gamma_{k,n}$ is the achievable SINR at time $ n $  given  by \eqref{sinr}.
		\begin{figure*}
			\begin{align}
				\gamma_{k,n}=\frac{\al_{k,n-K}^{2}p_{k}|\EE\{{\bh}_{k}^{\H}\bff_{k,n}\}|^{2}}{\al_{k,n-K}^{2}p_{k}\mathrm{Var}\left\{{\bh}_{k}^{\H}\bff_{k,n}\right\}		+
					\sum_{i\ne k}^{K}\al_{i,n-K}p_{i}\EE\left\{|\bh^\H_{k,n}\bff_{i,n}|^{2}\right\} 		+\bar{\al}_{k,n-K}^{2}p_{k}\EE\left\{|	\tilde{\bee}_{k,n}^{\H}\bff_{k,n}|^{2}\right\}+
					\sigma^{2}}.\label{sinr}
			\end{align}
			\line(1,0){490}
		\end{figure*}
	\end{proposition}
	\begin{proof}
		First, based on a similar approach to~\cite{Bjornson2015,Pitarokoilis2015},     the average achievable SE including the achievable SINR $ \gamma_{k,n} $ in the transmission phase is computed  for each $ n $. Next,   the average over these SEs is obtained as in \eqref{LowerBound}.

		Regarding the achievable SINR $ \gamma_{k,n} $,  given  the Gaussianity of the input symbols, it  is obtained by accounting for    a worst-case assumption for the computation of the mutual information~\cite[Theorem $1$]{Hassibi2003}. In particular, except for $ \EE\{\bh^\H_{k,n}\bff_{k,n}\} $ being the deterministic desired signal detected by UE $ k $, all others terms are treated as independent Gaussian noise with zero mean and  variance equal to the variance of interference plus noise.  
	\end{proof}
	
	The following analysis requires $ M $, $ L $, and $ K $ increase  but with a given bounded ratio as $ 0 < \lim \inf \frac{K
	}{M}\le \lim \sup \frac{K
	}{M}< \infty $  and $ 0 < \lim \inf \frac{L
	}{M}\le \lim \sup \frac{L
	}{M}< \infty $. Henceforth, this notation is denoted as $ M \to \infty$. Taking into account that, according to \eqref{innov2}, the available CSI at time $ n $ is $ \al_{k,n-K}\hat{\bh}_{k} $, the BS designs its RZF precoder as\footnote{Despite that RZF precoding is generally suboptimal,  it has been applied selected in the case of mMIMO in many works due to reasons of complexity and to provide closed-form expressions. In other words, it is common in mMIMO to consider a linear precoder between the MRT and RZF precoders. Hence, for these reasons, in this work, we have selected the more optimal precoder, i.e., RZF despite its complexity. Note that  RZF precoding is a very good choice compared to MRT, and the corresponding derivation of the rate together with its optimization require delicate manipulations, which raise the novelty of this work.}
	\begin{align}
		\bff_{k,n}=
		\al_{k,n-K}\sqrt{\lambda} \bSigma\hat{\bh}_{k}, \label{precoderRZF}
	\end{align}
	where $ \lambda $ is   a normalization parameter, which is obtained due to \eqref{PowerConstraint} as 
	\begin{align}
		\lambda=\frac{P_{\mathrm{max}}}{	\al_{k,n-K}^{2}\tr(\bP\hat{\bH}^{\H}\bSigma^{2}\hat{\bH})}. \label{eq:lamda} 
	\end{align} 
	Note that ${\bSigma}\triangleq	\left(\al_{k,n-K}^{2}\hat{\bH}\hat{\bH}^{\H} \!+\! \bZ + M \al \Id_M\right)^{-1}$, where $ \bZ \in \mathbb{C}^{M \times M} $ is  an arbitrary Hermitian non negative definite matrix, and $\al$ is a regularization scaled by $M$ to make expressions converge to a constant as $M \to \infty$.  As a result, by denoting $ \rho=\frac{P_{\mathrm{max}}}{	\sigma^{2}} $ the SNR at the downlink transmission phase, the SINR of UE $ k $ under RZF precoding can be written as \eqref{SINRwithPrecoding} at the top of the next page. For the sake of convenience, we denote $ S_{k,n} $ and $ I_{k,n} $ the numerator and denominator of \eqref{SINRwithPrecoding}, respectively, i.e., we have  $ \gamma_{k,n}=\frac{S_{k,n}}{I_{k,n}}$.
	\begin{figure*}
		\begin{align}
			\gamma_{k,n}=\frac{p_{k}|\EE\{{\bh}_{k}^{\H}\bSigma\hat{\bh}_{k}	 \}|^{2}}{p_{k}\mathrm{Var}\left\{{\bh}_{k}^{\H}\bSigma\hat{\bh}_{k}	 \right\}		+
				\sum_{i\ne k}^{K}p_{i}\frac{\al_{i,n-K}^{4}}{\al_{k,n-K}^{4}}\EE\left\{|\bh^\H_{k,n}\bSigma\hat{\bh}_{i}	 |^{2}\right\} 		+p_{k}\frac{\bar{\al}_{k,n-K}^{2}}{\al_{k,n-K}^{2}}\EE\left\{|	\tilde{\bee}_{k,n}^{\H}\bSigma\hat{\bh}_{k}	 |^{2}\right\}+
				\frac{\tr(\bP\hat{\bH}^{\H}\bSigma^{2}\hat{\bH})}{\al_{k,n-K}^{2}\rho}}\label{SINRwithPrecoding}.
		\end{align}
		\line(1,0){490}
	\end{figure*}
	
	Based on similar assumptions to \cite[Assump. A1-A3]{Hoydis2013} regarding the covariance matrices under study, we rely on the DE analysis to obtain the DE downlink SE of UE $ k $.\footnote{The DE analysis results in deterministic expressions, which make lengthy Monte-Carlo simulations unnecessary. Also, its results are tight approximations even for conventional systems with moderate  dimensions, e.g., an $ 8 \times 8 $ matrix~\cite{Couillet2011}. Hence, the following DE SINR in \eqref{DLdelayedCSIetaRZF} is of great practical importance.} 	The DE SINR  obeys to $\gamma_{k,n}-\bar{\gamma}_{k,n}\xrightarrow[M \rightarrow \infty]{\mbox{a.s.}}0$ while the deterministic SE of UE $k$ obeys to
	\begin{align}
		\mathrm{SE}_{k}-\mathrm{\overline{SE}}_{k}\xrightarrow[M \rightarrow \infty]{\mbox{a.s.}}0,\label{DeterministicSumrate}
	\end{align}
	where $\mathrm{\overline{SE}}_{k}= \frac{1}{T_{\mathrm{c}}}\sum_{n=K+1}^{T_{\mathrm{c}}}\log_{2}(1 + \bar{\gamma}_{k,n}) $ based on the dominated convergence and the continuous mapping theorem~\cite{Couillet2011}.

	\begin{Theorem}\label{theorem:DLagedCSIRZF}
		The downlink DE of the SINR of UE $k$ with RZF precoding at time $n$, accounting for correlated Rayleigh fading, imperfect CSI, and channel aging due to  UE mobility  is given by \eqref{DLdelayedCSIetaRZF}, where
		\begin{figure*}
			\begin{align}
				\!	\!\bar{\gamma}_{k,n} \!=\! \frac{ p_{k}\delta_{k}^{2}	 }{\displaystyle p_{k}\tilde{\delta}_{k}\!+\!	p_{k}\frac{\bar{\al}_{k,n-K}^{2}}{\al_{k,n-K}^{2}}\delta_{k}^{\mathrm{e}}\!+\!\sum_{i\ne k}\!p_{i}\frac{\al_{i,n-K}^{4}(1\!+\!\al_{k,n-K}^{2}\delta_{k})^{2}}{\al_{k,n-K}^{4}(1\!+\!\al_{i,n-K}^{2}{\delta_{i}})^{2}}Q_{ik}\!+\!	\frac{1}{M}\sum_{i=1}^{K}p_{i}\frac{(1\!+\!\al_{k,n-K}^{2}\delta_{k})^{2}	\delta_{i}^{\lambda}	 }{\rho(1+	\al_{i,n-K}^{2}\delta_{i})^{2}}}\label{DLdelayedCSIetaRZF}
			\end{align}
			\line(1,0){490}
		\end{figure*}
		\begin{align}
			\bar{\lambda}&=\frac{P_{\mathrm{max}}}{  \frac{1}{M}\sum_{i=1}^{K}p_{i}\frac{\al_{k,n-K}^{2}	\delta_{i}^{\lambda}	 }{(1+	\al_{i,n-K}^{2}\delta_{i})^{2}}},
		\end{align}
		$\delta_{k}= \frac{1}{M}\tr(\bPhi_{k} \bT ),
		{\tilde{\delta}}_{k}=\frac{1}{M^{2}}\tr \big((\bR_{k}-\bPhi_{k})\tilde{\bT}(\bPhi_{k})\big),
		\delta_{k}^{\mathrm{e}}=\frac{1}{M}\tr\big(\bR_{k}\tilde{\bT}(\bPhi_{k})\big)$, 	$\delta_{i}^{\lambda}=\frac{1}{M}\tr\big(\bPhi_{i}\tilde{\bT}(\Id_{M})\big)		$,
		$ \zeta_{ki}=\frac{1}{M^{2}}\tr\big(\bR_{k}\tilde{\bT}(\bPhi_{i})\big) $, 	 $ 	\mu_{ki}=\frac{1}{M^{2}}\tr\big(\bPhi_{k}\tilde{\bT}(\bPhi_{i})\big) $,
		$Q_{ik}= \zeta_{ki}\!+\!\frac{|\delta_{k}|^{2}\mu_{ki}}{\left(1\!+\!\al_{k,n-K}^{2}\delta_{k}\right)^{2}}\!-\!2\mathrm{Re}\left\{\!\!
		\frac{ \delta_{k}^{*}\mu_{ki}}{1\!+\!\al_{k,n-K}^{2}\delta_{k}}\right\}$, 
		with
		\begin{itemize}
			\renewcommand{\labelitemi}{$\ast$}
			\item $\bT\!=\!\Big(\!\displaystyle\sum_{i=1 }^{K} \frac{1}{\!M\left(1+\al_{i,n-K}^{2}\delta_{i}\right)}\bPhi_{i}\!+\!\al \Id_{M}\Big)^{-1}$,
			\item $ \tilde{\bT}(\bL)=\bT\bL\bPhi_{i}\bT+\sum_{i=1}^{K}\frac{\tilde{\delta}_{i}\bT\bPhi_{i}\bT}{M\left(1+\al_{i,n-K}^{2}\delta_{i}\right)^{2}}$, $ \bL=\bPhi_{i}, \Id_{M} $
			\item $ \tilde{\deltav}(\bL) \!\!=\left(\Id_{K}\!-\!\bF\right)^{\!-1}\bf$ with $ 	\left[\bF\right]_{k,i}\!=\!\frac{1}{M^{2}\left(1+ \al_{i,n-K}^{2}\delta_{i}\right)}\!\tr\!\left(\bPhi_{k}\bT\bL\bT\right) $, $ \bL=\bPhi_{i}, \Id_{M} $.
		\end{itemize} 
	\end{Theorem}
	\proof: Please see Appendix~\ref{theorem3}.\endproof

	\section{Sum SE Maximization}\label{SumSEMaximizationDesign}
	In this section, we focus on the  optimization of the sum SE of a RIS-assisted time-varying mMIMO  system with channel aging. Specifically, the sum rate maximization problem is described as
	\begin{subequations}
		\begin{align}
			(\mathcal{P}1)~~&\max_{\bTheta,\bp\ge 0} 	\;	\mathrm{\overline{SE}}=\sum_{i=1}^{K}\mathrm{\overline{SE}}_{k}
			\label{Maximization1} \\
			&~	\mathrm{s.t}~~~\;\!	|\tilde{\theta}_{l}|=1,~~ l=1,\dots,L,\label{Maximization2} \\
			&\;\quad\;\;\;\;\;\!\!~\!\sum_{i=1}^{K}p_{i}\le P_{\mathrm{max}},
			\label{Maximization3} 
		\end{align}
	\end{subequations}
	where  we have denoted the elements of $ \bTheta $ as $\tilde{\theta}_{l}= \exp\left(j \theta_{l}\right) $ for all $ l $ and the vector $ \bp=[p_{1}, \ldots,p_{K}]^{\T} $. The first constraint in \eqref{Maximization2} means that each RIS element induces a phase shift without any change on the amplitude of the incoming signal, while  constraint \eqref{Maximization3}  ensures that the BS transmit power is kept below the maximum power $ P_{\mathrm{max}} $.\footnote{The considered optimization problem allocates the available resources (power vector, RIS configuration) to maximize the total spectral efficiency (i.e, sum-rate). This is a well-known objective function when the total/aggregate rate of the network is the key design priority. It is worth also noting that the  sum capacity/rate is a fundamental performance metric (one  dimension) in multi-user networks from information theoretic perspective. To consider QoS per individual user, the consideration of other objective functions that take into account fairness (e.g. max-min rate) and/or minimum individual rates is required, which are beyond the scope of this paper and can be considered for future work.}
	
	The solution of the optimization problem $ 	(\mathcal{P}1) $ is challenging  due to its non-convexity and the  unit-modulus constraint concerning $ \tilde{\theta}_{l} $. To tackle this difficulty, we consider the alternating optimization (AO) technique, where $ \bTheta $ 	and  $ \bp $ are going to be solved separately and iteratively. In particular,  first,  we focus on finding  the optimum  $ \bTheta $ given  fixed $ \bp $. Next, we solve for $ \bp $ with  fixed $ \bTheta $. The iteration of this procedure, where the sum rate increases  at  each iteration step, continues until convergence  to  the optimum value  since the sum rate is upper-bounded  subject to the power constraint \eqref{Maximization3}.  Note that all computations take place at the BS.

	\subsection{RIS Configuration}
	The exploitation of the RIS potentials implies the optimization of the RBM towards maximum sum SE. The presence of the RBM appears inside the covariance matrices  in the DE achievable SINR in \eqref{DLdelayedCSIetaRZF}. Given that the logarithm function is monotonic, it is sufficient to maximize $ \bar{\gamma}_{k,n} $ instead of $\mathrm{\overline{SE}}_{k} $. Hence, by assuming infinite resolution phase shifters, the RBM optimization problem   is formulated as
	\begin{align}\begin{split}
			(\mathcal{P}2)~~~~~~~\max_{\bTheta} ~~~	&	\;	\mathrm{\overline{SE}}\\
			\mathrm{s.t}~~~&|\tilde{\theta}_{l}|=1,~~ l=1,\dots,L.
		\end{split}\label{Maximization} 
	\end{align}

	Although the  problem  $ 	(\mathcal{P}2) $ is  non-convex in terms of $ \bTheta $ and it is subject to  a unit-modulus constraint regarding $ \tilde{\theta}_{l} $, application of  the projected gradient 	ascent algorithm  can achieve a local optimal solution  by projecting the solution onto the closest feasible point at every step 	  until converging to a 	stationary point. Specifically, let the vector $ \bc_{t} =[\tilde{\theta}_{t,1}, \ldots, \tilde{\theta}_{t,L}]^{\T}$ include the phases at step  $ t $. The next step of the algorithm is described by
	\begin{align}
		\tilde{\bc}_{t+1}&=\bc_{t}+\tilde{\mu} \bq_{t},~
		\bc_{t+1}=\exp\left(j \arg \left(\tilde{\bc}_{t+1}\right)\right),\label{sol2}
	\end{align}
	where the parameter $ \tilde{\mu }$ denotes the step size and  $ \bq_{t}= \pdv{ \bar{\gamma}_{k,n}}{\bc_{t}^{*}} $  describes the  ascent direction at step $ t $, given below by Proposition \ref{Prop:optimPhase}.\footnote{Given that the feasible set is the unit circle, then any point  $ x $ should be   projected on this circle, i.e., it should be $ x/|x| $, which is equal to $ \exp(j \angle(x)) $.} Note that  the suitable step size  at each iteration is selected based on  the backtracking line search \cite{Boyd2004}. The solution is obtained by the  projection problem $ \min_{|\theta_{l} |=1, l=1,\ldots,L}\|\bc-\tilde{\bc}\|^{2} $ under the unit-modulus constraint. 	Algorithm \ref{Algoa1} presents the overview of this procedure.  For the sake of convenience, we denote  the partial derivative with respect to $ \bc_{t}^{*} $  by $ (\cdot)' $.
	
	\begin{algorithm}
		\caption{Projected Gradient Ascent Algorithm for the RIS Design}
		1.				 \textbf{Initialisation}: $ \bc_{0} =\exp\left(j\pi/2\right)\one_{L}$, $ \bTheta_{0}=\diag\left(\bc_{0}\right) $, $\bar{\gamma}_{k,n}=f\left(\bTheta_{0}\right) $ given by \eqref{DLdelayedCSIetaRZF}; $ \epsilon>0 $\\
		2. \textbf{Iteration} $ t $: \textbf{for} $ t=0,1,\dots, $ do\\
		3. $\bq_{t}=   \pdv{\bar{\gamma}_{k,n}}{\bc_{t}^{*}}  $, where $\pdv{\bar{\gamma}_{k,n}}{\bc_{t}^{*}} $ is given by Proposition \ref{Prop:optimPhase};\\
		4. \textbf{Find} $ \tilde{\mu} $ by backtrack line search $( f\left(\bTheta_{0}\right),\bq_{t},\bc_{t})$ \cite{Boyd2004};\\
		5. $ \tilde{\bc}_{t+1}=\bc_{t}+\tilde{\mu} \bq_{t} $;\\
		6. 	$ \bc_{t+1}=\exp\left(j \arg \left(\tilde{\bc}_{l+1}\right)\right) $; $ \bTheta_{l+1}= \diag\left(\bc_{l+1}\right) $;\\
		7. $ \bar{\gamma}_{k,n}^{t+1}=f\left(\bTheta_{t+1}\right) $;\\
		8. \textbf{Until} $ \| \bar{\gamma}_{k,n}^{t+1}- \bar{\gamma}_{k,n}^{t}\|^{2} <\epsilon$; \textbf{Obtain} $ \bTheta^{\star}=\bTheta_{t+1}$;\\
		9. \textbf{end for}\label{Algoa1}
	\end{algorithm}

	\begin{proposition}\label{Prop:optimPhase}
		The derivative of $ \bar{\gamma}_{k,n} $ with respect to $ \bc_{t}^{*} $ is given  by
		\begin{align}
			\!\bar{\gamma}_{k,n}' =\frac{S_{k}'I_{k}-S_{k}I_{k}'}{I_{k}^{2}},\label{gam1}
		\end{align}
		where 
		\begin{align}
			S_{k}'&= 2p_{k}\delta_{k}\delta_{k}', \label{sg11}\\
			I_{k}'&=\displaystyle\tilde{\delta}_{k}'\!+\!(\delta_{k}^{\mathrm{e}})'\!+\!\sum_{i\ne k}\!\frac{\al_{i,n-K}^{2}(1\!+\!\al_{k,n-K}^{2}\delta_{k})^{2}}{\al_{k,n-K}^{2}(1\!+\!\al_{i,n-K}^{2}{\delta_{i}})^{2}}Q_{ik}'	\nn\\
			\!&+\frac{1}{M}\sum_{i=1}^{K}p_{i}\frac{2(1\!+\!\al_{k,n-K}^{2}\delta_{k})\delta_{k'}	\delta_{i}^{\lambda}+(1\!+\!\al_{k,n-K}^{2}\delta_{k})^{2}	(\delta_{i}^{\lambda})'	 }{(1+	\al_{i,n-K}^{2}\delta_{i})^{2}}\nn\\
			\!&+\frac{1}{M}\sum_{i=1}^{K}p_{i}\frac{\al_{i,n-K}^{2}(1\!+\!\al_{k,n-K}^{2}\delta_{k})^{2}\rho\delta_{i}'	\delta_{i}^{\lambda}	 }{(1+	\al_{i,n-K}^{2}\delta_{i})^{3}}\nn\\
			\!&+\sum_{i\ne k}\!\frac{2 \al_{i,n-K}^{2}\al_{k,n-K}^{2}\delta_{k}'(1\!+\!\al_{k,n-K}^{2}\delta_{k})}{\al_{k,n-K}^{2}\rho(1\!+\!\al_{i,n-K}^{2}{\delta_{i}})^{3}}Q_{ik}\!\nn\\
			&-\!\sum_{i\ne k}\!\frac{3\al_{i,n-K}^{4}\delta_{i}'(1\!+\!\al_{k,n-K}^{2}\delta_{k})^{2}}{\al_{k,n-K}^{2}(1\!+\!\al_{i,n-K}^{2}{\delta_{i}})^{5}}Q_{ik}\label{intDeriv}
		\end{align}
		with the auxiliary variables  given by
		\begin{align}
			\!\!	\!\!\!\! &\tilde{\delta}_{k}'\!=\!\frac{1}{M}\tr\!\big(\!(\bR_{k}'-\bPhi_{k}')\tilde{\bT}(\bPhi_{k})\!+\!(\bR_{k}-\bPhi_{k})\tilde{\bT}'(\bPhi_{k})\!\big),\label{sg41}\\
			( &\delta_{k}^{\mathrm{e}})'\!=\!\frac{1}{M}\tr\big(\bR_{k}'\tilde{\bT}(\bPhi_{k})+\bR_{k}\tilde{\bT}'(\bPhi_{k})\big),\label{sg51}\\
			& \bT'=-\bT(\bT^{-1})'\bT,	\\
			&	(\bT^{-1})'=\frac{1}{M}\displaystyle \sum_{i=1 }^{K}\frac{\bPhi_{i}'\left(1+\al_{i,n-K}^{2}\delta_{i}\right)-\al_{i,n-K}^{2}\bPhi_{i}\delta_{i}'}{(1+\al_{i,n-K}^{2}\delta_{i})^{2}},\label{sg31}\\
			&\tilde{\bT}'(\bPhi_{k})=\bT'\bPhi_{k}\bT+\bT\bPhi_{k}'\bT+\bT\bPhi_{k}\bT'\nn\\
			&+\frac{1}{M}\sum_{i=1}^{K}\frac{\tilde{\delta}_{i}'\bT\bPhi_{i}\bT+\tilde{\delta}_{i}(\bT'\bPhi_{k}\bT+\bT\bPhi_{k}'\bT+\bT\bPhi_{k}\bT')}{(1+\al_{i,n-K}^{2}\delta_{i})^{3}}\nn\\
			&-\frac{1}{M}\sum_{i=1}^{K}\frac{2\al_{i,n-K}^{2}\tilde{\delta}_{i}\bT\bPhi_{i}\bT \delta_{i}'}{(1+\al_{i,n-K}^{2}\delta_{i})^{}},\\
			& 	Q_{ik}'= \zeta_{ki}'\!+\!\frac{\delta_{k}\delta_{k}'\mu_{ki}+|\delta_{k}|^{2}\mu_{ki}}{\big(1\!+\!\al_{k,n-K}^{2}\delta_{k}\big)^{\!2}}\!+\!\frac{2\al_{k,n-K}^{2}|\delta_{k}|^{2}\mu_{ki}\delta_{k}'}{\big(1\!+\!\al_{k,n-K}^{2}\delta_{k}\big)^{\!3}}\nn\\
			&-\!2\mathrm{Re}\left\{\!\!
			\frac{ (\delta_{k}^{*})'\mu_{ki}+\delta_{k}^{*}\mu_{ki}'}{1\!+\!\al_{k,n-K}^{2}\delta_{k}}+\!\!\!
			\frac{ \al_{k,n-K}^{2}\delta_{k}^{*}\mu_{ki}\delta_{k}'}{\big(1\!+\!\al_{k,n-K}^{2}\delta_{k}\big)^{\!2}}\right\},\\
			& 	\zeta_{ki}'=\frac{1}{M^{2}}\tr\big(\bR_{k}'\tilde{\bT}(\bPhi_{i})+\bR_{k}\tilde{\bT}'(\bPhi_{i})\big), \\
			&\mu_{ki}'=\frac{1}{M^{2}}\tr\big(\bPhi_{k}'\tilde{\bT}(\bPhi_{i})+\bPhi_{k}\tilde{\bT}'(\bPhi_{i})\big),\\
			& 	\bar{\lambda}'=\frac{K \big( \frac{1}{M} \tr \left(\frac{\Zm}{M} +\al \Id_M\right) \tilde{\bT}'(\Id_{M})-\frac{1}{M}\tr\Tm'\big)}{ \big(\frac{1}{M}\tr\Tm - \frac{1}{M} \tr \left(\frac{\Zm}{M} +\al \Id_M\right) \tilde{\bT}(\Id_{M})\big)^{2}}.
		\end{align}

	\end{proposition}
	\proof The proof  is given in Appendix~\ref{optimPhase}.\endproof
	
	The RBM beamforming design is based on the gradient ascent , which results in a significant advantage because the gradient ascent is derived in a closed-form. This method comes with low computational complexity because it consists of simple matrix operations. Specifically, the complexity of Algorithm \ref{Algoa1} is $ \mathcal{O}\left(MN^{2}+N+M\right) $, which consists of the fundamental system parameters  $ M $ and $N $ with the number of RIS elements having the higher (square) impact.
	
	\subsection{Power Allocation}
	Given  a fixed RBM $ \bTheta $, the objective is the maximization of the sum SE with respect to $ \bp $. Specifically, we have 
	\begin{align}\begin{split}
			(\mathcal{P}3)~~~~~~~\max_{\bp\ge 0} ~~~	&			\mathrm{\overline{SE}}\\
			\mathrm{s.t}~~\;\!&\sum_{i=1}^{K}p_{i}\le P_{\mathrm{max}},
		\end{split}\label{Maximization} 
	\end{align}
	where  	$\mathrm{\overline{SE}}$ is given by \eqref{Maximization1}. 	This problem is not convex but a local optimal solution can be obtained by using a weighted minimum mean square error (MMSE) reformulation of the sum SE maximization. The SINR $\bar{\gamma}_{k,n} $ in \eqref{DLdelayedCSIetaRZF} can be described as a function of the downlink power coefficients given by the vector $ \bp $ as
	\begin{align}
		\bar{\gamma}_{k,n} =\frac{p_{k} q_{k}}{\bc^{\T}\bp},\label{SINR1}
	\end{align}
	where  $\bc\!=\![c_{1}, \ldots,c_{K} ]^{\T}  $ with 
	\begin{align}
		q_{k}&\!=\!\delta_{k}^{2},\! \\
		c_{k}&\!=\!\tilde{\delta}_{k}\!+\!	\frac{\bar{\al}_{k,n-K}^{2}}{\al_{k,n-K}^{2}}\delta_{k}^{\mathrm{e}}+\!\frac{1}{M}	\delta_{i}^{\lambda}\!,~\forall g \\
		c_{gi}&=\frac{\al_{i,n-K}^{4}(1\!+\!\al_{k,n-K}^{2}\delta_{k})^{2}}{\al_{k,n-K}^{4}(1\!+\!\al_{i,n-K}^{2}{\delta_{i}})^{2}}Q_{ik}\nn\\
		&+\frac{1}{M}\frac{(1\!+\!\al_{k,n-K}^{2}\delta_{k})^{2}	\delta_{i}^{\lambda}	 }{(1+	\al_{i,n-K}^{2}\delta_{i})^{2}},~\forall g, ~\forall i\ne g.
	\end{align}
	
	The MMSE reformulation includes writing  the SINR in \eqref{SINR1} in terms of a single‐input and single‐output (SISO) channel that is described as
	\begin{align}
		\tilde{y}_{k}=\sqrt{p_{k} q_{k}}s_{k}+\sum_{i=1}^{K}\sqrt{p_{i}c_{i}}s_{i},
	\end{align}
	where $ \tilde{y}_{k} $ is the received signal, and $ s_{i }\in \mathbb{C} $ expresses the normalized and independent random data signal with $ \EE\{|s_{i}|^{2}=1\} $. The receiver can compute an estimate $\hat{s}_{k}=v_{k}^{*} 	\tilde{y}_{k} $ of the desired signal $ s_{k} $ by minimizing the  MSE 	$ e_{k}(\bp,v_{k}) =[|\hat{s}_{k}-s_{k}|^{2}]$, where  $ v_{k} $ as a scalar combining coefficient. Specifically, the MSE is written as
	\begin{align}
		e_{k}(\bp,v_{k})=v_{k}^{2}\left(p_{k}q_{k}+\bc^{\T}\bp\right)-2 v_{k}\sqrt{q_{k}P_k}+1.\label{mse1}
	\end{align}
	
	The coefficient $ 	v_{k} $, which minimizes the MSE 	$ e_{k}(\bp,v_{k})$ for a given $ \bp $, is given by 
	\begin{align}
		v_{k}=\frac{\sqrt{p_{k} q_{k}}}{p_{k}q_{k}+\sum_{i=1}^{K}p_{i}c_{i}}.
	\end{align}

	Substituting $ v_{k} $ into \eqref{mse1}, the MSE becomes $ 1/\left(1+	\bar{\gamma}_{k,n} \right) $. According to the  weighted MMSE method, we introduce the auxiliary weight parameter $ d_{k}\ge 0 $ for the MSE $  e_{k}$ and focus on the solution of the following  optimization problem
	\begin{align}\begin{split}
			(\mathcal{P}4)~\min_{\substack{\bp\ge 0,\\
					\{v_{k}, d_{k}\ge 0: k=1,\ldots,K\}}} 	&		\!\sum_{i=1}^{K}\!d_{i} e_{i}(\bp,\bv_{i})-\ln(d_{1})\\
			\mathrm{s.t}~~\;\!&\sum_{i=1}^{K}p_{i}\le P_{\mathrm{max}}.
		\end{split}\label{Maximization10} 
	\end{align}

	It is worthwhile to mention that  problems  $ 	(\mathcal{P}3) $ and $ 	(\mathcal{P}4) $ are equivalent because they have the same global optimal solution. The equivalence  relies on the fact that the optimal $ d_{k} $ in \eqref{Maximization10} is $ 1/ e_{k}=(1+		\bar{\gamma}_{k,n})$. The benefit of  the reformulation results in the following lemma, which is  adapted based on \cite[Th. 3]{Shi2011}.
	
	\begin{lemma}
		The block descent coordinate algorithm, described as Algorithm \ref{Algoa2} below, converges to a local optimum of $ 	(\mathcal{P}4) $ by means of AO among three blocks of variables being $ \{v_{k}:  k=1,\ldots,K\} $, $ \{d_{k}:  k=1,\ldots,K\} $, and $ \bp $.
	\end{lemma}

	\begin{algorithm}
		\caption{Block coordinate descent algorithm for solving $ 	(\mathcal{P}4) $}
	1.				 \textbf{Initialisation}: Set $ \bp=\frac{P_{\mathrm{max}}}{K}\one_{K} $ (arbitrary value) and the solution accuracy $ \epsilon >0 $, \\
	2. \textbf{while} the objective function in \eqref{Maximization10} is not improved more than $ \epsilon $ \textbf{do}\\
	3. $ 	v_{k}=\frac{\sqrt{p_{k} q_{k}}}{p_{k}q_{k}+\sum_{i=1}^{K}p_{i}c_{i}}, $ $ i=1,\ldots, K $\\
	4. $ d_{k}=1/e_{k}(\bp,v_{k}), $ $ k=1,\ldots, K $\\
	5. Solve the following problem for the current values of $ v_{k} $ and $ d_{k} $:\\\vskip-12mm
	\begin{align}\begin{split}
			(\mathcal{P}5)~~~~~~~\min_{\bp \ge 0} ~~~	&		\sum_{i=1}^{K}d_{i} e_{i}(\bp,\bv_{i})\\
			\mathrm{s.t}~~\;\!&\sum_{i=1}^{K}p_{i}\le P_{\mathrm{max}},\\
		\end{split}\label{Maximization11} 
	\end{align}	\vskip-2mm
	6. Update $ \bp $ by the obtained solution to \eqref{Maximization11}\\
	7. \textbf{end while}\\
	8. \textbf{Output:} $ \bp^{\star} $\label{Algoa2}
\end{algorithm}

Algorithm \ref{Algoa2}  delineates the whole operation for the power allocation.  Step $ 5 $  includes a subproblem that has to be solved in every iteration. Fortunately, its solution can be obtained in closed-form with    $ \sqrt{p_{k}}, k=1\ldots, K $  treated as optimization variables by   decomposing  the problem into $ K$ independent subproblems. In particular, the solution is given by 
\begin{align}
	p_{k}=\min\!\bigg(\!\!P_{\mathrm{max}},\frac{q_{k}d_{k}^{2}v_{k}^{2}}{\big(q_{k}d_{k}v_{k}^{2}+\sum_{i=1}^{K}d_{i}v_{i}^{2}c_{i}\big)^{\!2}}\!\!\bigg).
\end{align}


\begin{remark}
	The proposed algorithms, i.e., Algorithms \ref{Algoa1} and \ref{Algoa2} converge quickly and have low computation complexity. Moreover, given that both algorithms achieve  a local optimum,  and that the overall algorithm is based on AO,  the final solution corresponds to a local optimum, which means that different initializations  will result  in different solutions,  as will be shown below in Sec. \ref{Numerical}.
\end{remark}

	The power allocation presents a similar complexity to the RBM design since similar matrix operations take place in Algorithm 2, i.e., its complexity is $ \mathcal{O}\left(MN^{2}+N+M\right) $.

\section{Numerical Results}\label{Numerical}
In this section, we elaborate on the numerical results concerning the downlink sum SE of RIS-assisted mMIMO systems with  channel aging, imperfect CSI, and correlated Rayleigh fading. MC simulations, portrayed by "\ding{53}" marks, verify our analysis for $n \to \infty$ even for finite (conventional) system dimensions.  This property has been already observed in previous works relying on DE analysis~\cite{Couillet2011,Hoydis2013,Papazafeiropoulos2015a}. For the sake of comparison, in certain cases, we also depict the scenarios with MRT precoding, static UEs, and the absence of RIS.

The simulation setup consists of a a single cell with radius  $R = 1000$ meters. In the center of the cell, there is a BS, which consists of a  uniform linear array (ULA) with $ M =100$ antennas that serve   $ K = 20 $ moving UEs while the communication is assisted  by an RIS assembled by a uniform planar array (UPA) of $ L=100 $ elements.  The length of the uplink training duration is $\tau=K$ symbols and the pilot transmit power is $p_{\mathrm{p}}=6~\mathrm{dBm}$. Assuming that the coherence time and bandwidth are $T_{\mathrm{c}}=2~\mathrm{ms}$ and $B_{\mathrm{c}}=100~\mathrm{kHz}$, respectively, i.e., the coherence block consists of $ 200 $ channel uses. Also,   $ \sigma^2=-174+10\log_{10}B_{\mathrm{c}} $. For the sake of exposition, we assume that the temporal correlation coefficient $ \al_{k,n} $ is the same across all UEs.

The correlation matrices $ \bR_{\mathrm{BS},k} $ and $ \bR_{\mathrm{RIS},k} $ are  computed according to \cite{Hoydis2013} and \cite{Bjoernson2020}, respectively. The size of each RIS element, required by the latter, is given by $ d_{\mathrm{H}}\!=\!d_{\mathrm{V}}\!=\!\lambda/4 $. In addition, the  path-losses corresponding to  the BS-to-RIS and RIS-to-UE $ k $ links are given by \cite{Wu2019,Kammoun2020}
\begin{align}
	\beta_{1}=\frac{C_{1}}{r_{1}^{\al_{1}}},~~~\beta_{2,k}=\frac{C_{2}}{r_{2,k}^{\al_{2}}},
\end{align}
where $ \al_{i} $ and $ r_{i} $ for $ i=1,2 $ are the path-loss exponent and distance for link $ i $. Note that  $ C_{1}=26 $ dB and $ C_{2}=28 $ dB, which describe  the path-losses at a reference distance of $ 1$ m  \cite{Bjoernson2019b}. In the case of the LoS link, we consider the same value for $ \beta_{\mathrm{d},k} $ as for $ \beta_{2,k} $ but with an additional penetration loss of $ 15~\mathrm{dB} $. The choice of these values  has been made based on the 3GPP Urban Micro (UMi) scenario from TR36.814 for a carrier frequency of $ 2 $ GHz and a noise level of $ -80 $ dBm. In particular,  the path-losses for links $ 1 $ and $ 2 $ are generated based on the LOS and NLOS versions \cite{Access2010}.  Specifically, we have $ \al_{1} =2.2$ and $ \al_{2} =3.67$ while $ r_{1} =8~\mathrm{m}$ and $ r_{2} =60~\mathrm{m}$. The Doppler spread, corresponding to a relative velocity of  $135$ km/h between the BS and the UEs, is $f_{\mathrm{D}}=250~\mathrm{Hz}$.    Moreover, if we assume that the   bandwidth is ${W}=20\mathrm{MHz}$, the symbol time is $T_s=1/(2\mathrm{W})=0.025~\mathrm{\mu s}$.

\begin{figure}[!h]
	\begin{center}
		\includegraphics[width=0.9\linewidth]{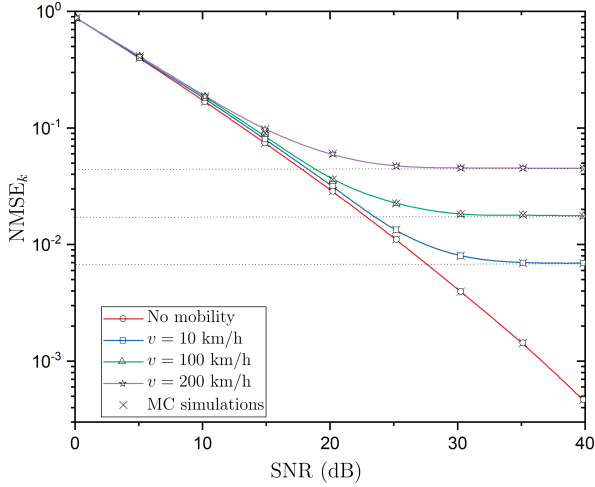}
		\caption{\footnotesize{NMSE of UE $ k $ versus the SNR of an RIS-assisted MIMO system with imperfect CSI ($ M=100 $, $ L=100 $, $ K=20 $) for different UE velocities in the cases of uncorrelated fading at the RIS (Analytical results and MC simulations). }}
		\label{Fig1}
	\end{center}
\end{figure}

Fig. \ref{Fig1} depicts the relative estimation error per channel element, i.e., the NMSE versus the downlink SNR for different velocities. Certain noise floors appear  as $ p \to \infty $. 
Moreover, we illustrate the result corresponding to no mobility, which decreases without bound. In addition, we  observe that the error floors take larger values with increasing  mobility (channel aging). It is shown that the NMSE saturates after $ 25~\mathrm{dB}$ for high velocity while it saturates much later for slow velocity. Note that these results correspond to uncorrelated fading at the RIS  to avoid any RBM optimization because, in this case, the NMSE is independent of $ \bTheta $
(Remark~\ref{rem1}). However, below, in the case of the sum SE, the optimization  with respect to $ \bTheta $ is applied.

\begin{figure}[!h]
	\begin{center}
		\includegraphics[width=0.9\linewidth]{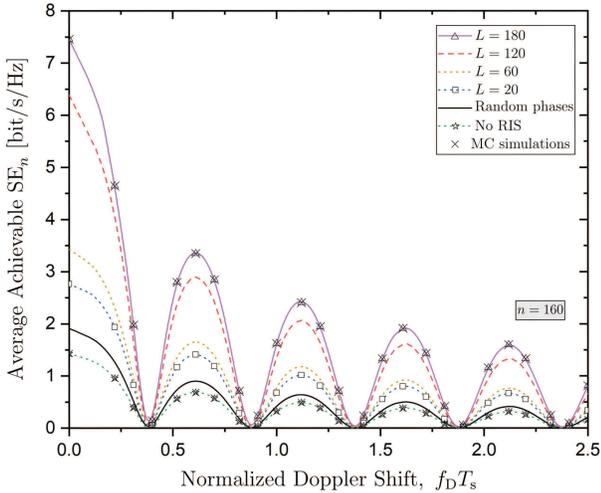}
		\caption{\footnotesize{Downlink achievable sum SE with RZF precoding of an RIS-assisted MIMO system with channel aging versus the normalized Doppler shift $ f_{\mathrm{D}}T_{\mathrm{s}}  $ for different
				number of RIS elements $ L $ ($ M=100 $, $ K=20 $)  (Analytical results and MC simulations). }}
		\label{Fig2}
	\end{center}
\end{figure}

In Fig \ref{Fig2}, we depict the achievable average  $ \mathrm{SE}_{n} $ versus the normalized Doppler shift $ f_{\mathrm{D}}T_{\mathrm{s}} $ for $ n=160 $ while varying the number of RIS elements, i.e., $ L=20, 60, 120,180 $. The ripples  are caused because of the dependence from the correlation coefficient,  which takes the form of the Bessel function of zeroth order, i.e., $\al_{k,n-K}^{2}=\mathrm{J}_{0}^{2}(2 \pi f_{\mathrm{D}}T_{\mathrm{s}}n)  $. It is shown that by increasing $f_{\mathrm{D}}T_{\mathrm{s}}$,  the average sum SE decreases. An increase in $ f_{\mathrm{D}}T_{\mathrm{s}} $ is equivalent to an increase in velocity or duration of the symbol time.  Especially, when $ f_{\mathrm{D}}T_{\mathrm{s}}\approx 0.39 $, the sum rate becomes almost zero for the first time, and then the magnitude fluctuates while tending to zero. Moreover, a larger RIS in terms of number of elements results in larger average sum SE but the curves keep the same shape. Even the zeros of the sum SE appear at the same normalized Doppler shifts. Hence, the presence of a RIS improves the performance even in channel aging conditions.
\begin{figure}[!h]
	\begin{center}
		\includegraphics[width=0.9\linewidth]{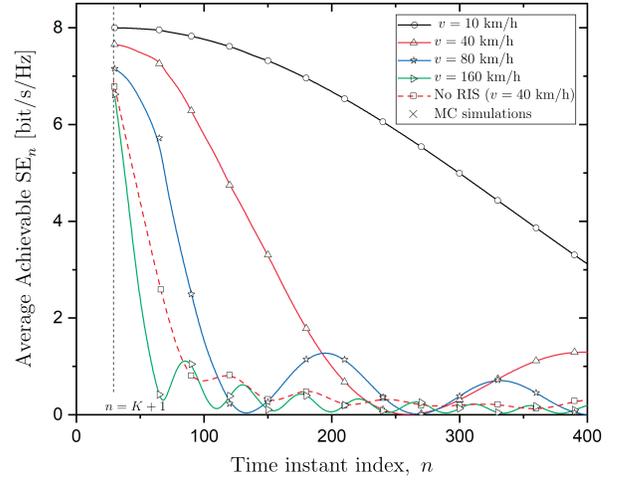}
		\caption{\footnotesize{Downlink achievable sum SE with RZF precoding of an RIS-assisted MIMO system with channel aging versus the time instant index $ n $ for different
				velocities $ v $ ($ M=100 $, $ L=100 $, $ K=20 $)  (Analytical results and MC simulations).  }}
		\label{Fig3}
	\end{center}
\end{figure}

Fig. \ref{Fig3} shows the achievable average  $ \mathrm{SE}_{n} $ versus the time instant index $ n $ for different UEs velocities. The time index starts at $ n=K+1 $, where the data transmission begins. As can be seen, at a given time instant, $ \mathrm{SE}_{n} $ decreases with UE mobility.  Moreover, as the velocity increases, the first zero position moves to the left, which suggests that the length of  coherence time should be designed to account for the impact of channel aging. Also, we depict the scenario with no RIS in the case $ v=40~\mathrm{km/h} $, which appears lower SE. 

\begin{figure}[!h]
	\begin{center}
		\includegraphics[width=0.9\linewidth]{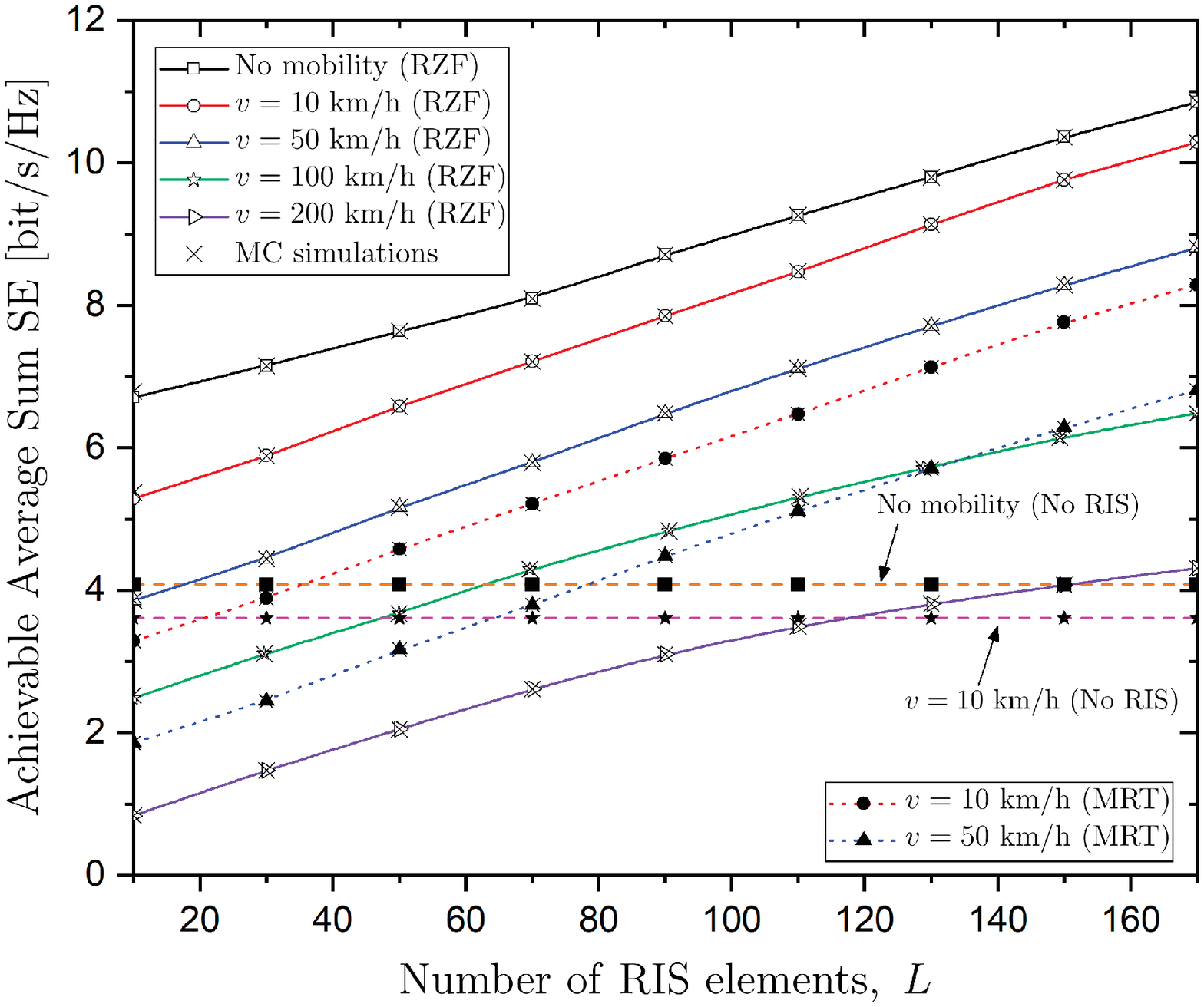}
		\caption{\footnotesize{Downlink achievable sum SE with RZF precoding of an RIS-assisted MIMO system with channel aging versus  the number of RIS elements $ L $ for different
				velocities $ v $ ($ M=100 $,  $ K=20 $)  (Analytical results and MC simulations).  }}
		\label{Fig4}
	\end{center}
\end{figure}

Fig. \ref{Fig4} illustrates the achievable downlink sum  $ \mathrm{\overline{SE}} $ versus the number of RIS elements $ L $ for varying UE mobility in terms of the relative velocity $ v $. Obviously, $ \mathrm{\overline{SE}} $ increases with $ L $ as expected, but an increase in velocity degrades the performance. The larger the velocity becomes, the larger the loss in the sum SE is observed. However, an increase with respect to the number of elements still allows for improving the SE despite the degradation due to mobility.  Also, we show the results corresponding to static UEs and no RIS that correspond to the horizontal lines, which are parallel to the $ x $-axis. For the sake of further comparison, we illustrate the performance of the simpler MRT precoding under channel aging conditions (dotted lines), which increases with $ L $ but is inferior compared to RZF.

\begin{figure}[!h]
	\begin{center}
		\includegraphics[width=0.9\linewidth]{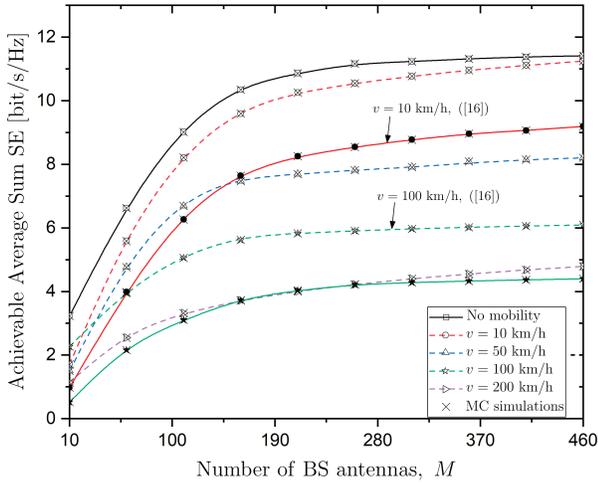}
		\caption{\footnotesize{Downlink achievable sum SE with RZF precoding of an RIS-assisted MIMO system with channel aging versus  the number BS antennas $ M $ for different
				velocities $ v $ ($ L=100 $, $ K=20 $)  (Analytical results and MC simulations). }}
		\label{Fig5}
	\end{center}
\end{figure}

Fig. \ref{Fig5} shows also the impact of channel aging in terms of the achievable sum SE with respect to the number of BS antennas $ M $ for different velocities. Generally, the downlink sum rate
presents  a rise with $ M $. Notably, we observe that when the velocity increases, the sum rate decreases while an increase in the number of BS antennas increases the performance. For the sake of comparison, we have included the solid lines in terms of simulation corresponding to the SE with CE performed according to \cite{Nadeem2020}   when $  v=10~\mathrm{km/h}$ and $  v=100~\mathrm{km/h} $. Although the CE in \cite{Nadeem2020} does not account for channel aging, the achievable sum SE is much lower than the proposed method due to its high training overhead.
\begin{figure}[!h]
	\begin{center}
		\includegraphics[width=0.9\linewidth]{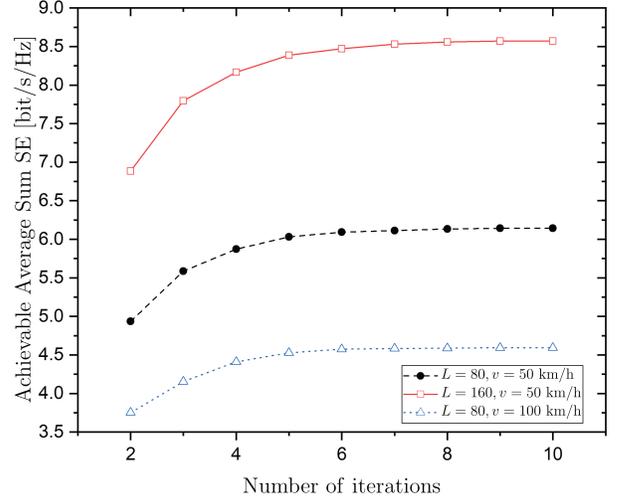}
		\caption{\footnotesize{Downlink achievable sum SE with RZF precoding of an RIS-assisted MIMO system with channel aging versus  the  number of iterations for different
				velocities $ v $ and number of RIS elements $ L $ ($ M=100 $, $ K=20 $)  (Analytical results and MC simulations).}}
		\label{Fig6}
	\end{center}
\end{figure}

In Fig. \ref{Fig6}, we depict the convergence of the proposed entire algorithm, which is founded on AO and includes the subproblems of RBM and transmit power optimizations. In particular, the horizontal axis corresponds to the number of iterations. As can be observed, the convergence is achieved quite fast for the cases under study, i.e., about $ 6 $ iterations are required. In the same figure, we have shown how the increase in the number of RIS elements and the channel aging affect the convergence. As $ L $ increases, more iterations are required to reach convergence because the number of optimization variables is larger. However, channel aging does not affect the rate of convergence but only the rate since the number of optimization variables  does not change.
\begin{figure}[!h]
	\begin{center}
		\includegraphics[width=0.9\linewidth]{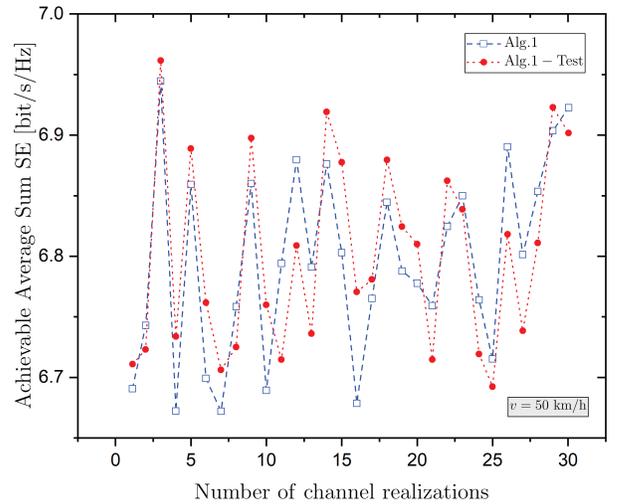}
		\caption{\footnotesize{Downlink achievable sum SE with RZF precoding of an RIS-assisted MIMO system with channel aging versus $ 30 $ channel realizations ($ M=100 $, $ L=100 $, $ K=20 $)  (Analytical results and MC simulations).}}
		\label{Fig7}
	\end{center}
\end{figure}

In addition, Fig. \ref{Fig7} assesses the dependence on the choice of the  initial points. Specifically,  given that the optimization problem \ref{Maximization3} is non-convex, its solution is contingent on the initialization. Hence, we consider  $30 $ channel realizations and $ v=50~\mathrm{km/h} $. The overall algorithm, consisted of Algorithms $ 1 $ and $ 2 $ is initialized with $ \bc_{0} =\exp\left(j\pi/2\right)\one_{L}$ and $ \bp=\frac{P_{\mathrm{max}}}{K}\one_{K}  $. The line corresponding to "Alg. 1-Test"  considers the best initial point out of $ 100 $ random initial points for each channel instance. As observed, different initial points lead to different solutions. However,  the sum SE in both cases is almost identical, which indicates that the selected initialization regarding the phase shifts and the transmit power  is a good choice. 
\section{Conclusion} \label{Conclusion} 
In this paper, we studied the impact of channel aging on RIS-assisted mMIMO systems by taking additionally into account the effects of spatial correlation and imperfect CSI. Having introduced channel aging and correlated Rayleigh fading not only during the data transmission phase but also in the uplink training phase, we obtained the effective channel estimates and derived the DE achievable sum SE with RZF in closed form. Next, we provided its maximization with respect to the RIS phase shifts and power budget constraints by applying an efficient AO algorithm based on statistical CSI that reduces  both the computational complexity and the feedback overhead. Hence, we illustrated the impact of channel aging and how its interconnection with other fundamental parameters affects performance. For instance, a suitable selection of the number of RIS elements and the length of the frame duration can mitigate channel aging. Notably, this work suggests interesting directions for future research with the study of wideband systems being of prominent significance. 
\begin{appendices}
	\section{Proof of Theorem~\ref{theorem:DLagedCSIRZF}}\label{theorem3}
	Regarding the  term in the desired signal power given by $ S_{k,n} $, we have
	\begin{align}
		\EE\{{\bh}_{k}^{\H}\bff_{k,n}\}&=	\al_{k,n-K} \EE\{(\hat{\bh}_{k}^{\H}+\tilde{\bh}_{k}^{\H}) \bSigma_{k}\hat{\bh}_{k}\}\label{ds0}\\
		&\asymp	\al_{k,n-K}\EE\left\{\frac{\hat{\bh}_{k}^{\H}	 \bSigma_{k}\hat{\bh}_{k}}{1+	\al_{k,n-K}^{2}\hat{\bh}_{k}^{\H} \bSigma_{k}\hat{\bh}_{k}}\right\}\label{d1}\\
		&\asymp	\al_{k,n-K}\frac{\frac{1}{M}\tr(\bPhi_{k} \bT )	 }{1+	\al_{k,n-K}^{2}\frac{1}{M}\tr(\bPhi_{k} \bT )}\label{d3}\\
		&=\frac{	\al_{k,n-K}\delta_{k}	 }{1+	\al_{k,n-K}^{2}\delta_{k}},\label{d2}
	\end{align}
	where $ {\bSigma}_k $ is defined as
	\begin{align}
		{\bSigma}_k &=\big(\al_{k,n-K}^{2}\hat{\bH}_{k}\hat{\bH}_{k}^{\H} -	\al_{k,n-K}^{2}\hat{\bh}_{k}\hat{\bh}_{k}^{\H}+ \bZ + \al M\Id_M\big)^{-1}.\nn
	\end{align}
	In \eqref{d1}, we have applied the matrix inversion lemma. In \eqref{d3}, we have used  \cite[Lem. 14.3]{Bai2010} known as rank-1 perturbation lemma, \cite[Lem. B.26]{Bai2010}, and \cite[Theorem 1]{Wagner2012}. Also, we have set $\delta_{k}= \frac{1}{M}\tr(\bPhi_{k} \bT )	 $.

	For the derivation of the DE of the normalization parameter $\lambda$, we focus on the denominator, and we result in
	\begin{align}
	&	\tr(\bP\hat{\bH}^{\H}\bSigma^{2}\hat{\bH})=\sum_{i=1}^{K}p_{i}\hat{\bh}_{i}^{\H}\bSigma^{2}\hat{\bh}_{i}  \\
		&= \frac{1}{M}\sum_{i=1}^{K}p_{i}\frac{\hat{\bh}_{i}^{\H}	 \bSigma_{i}^{2}\hat{\bh}_{i}}{(1+	\al_{i,n-K}^{2}\hat{\bh}_{i}^{\H} \bSigma_{i}\hat{\bh}_{i})^{2}}\label{lamda1}\\
		&\asymp \frac{1}{M}\sum_{i=1}^{K}p_{i}\frac{\frac{1}{M}\tr\big(\bPhi_{i} \tilde{\bT}(\Id_M)\big )	 }{(1+	\al_{i,n-K}^{2}\frac{1}{M}\tr(\bPhi_{i} \bT ))^{2}}\label{lamda2}\\
		&=\frac{1}{M}\sum_{i=1}^{K}p_{i}\frac{	\delta_{i}^{\lambda}	 }{(1+	\al_{i,n-K}^{2}\delta_{i})^{2}},\label{desired1}
	\end{align}
	where, in \eqref{lamda1}, we have applied the matrix inversion lemma twice. In \eqref{lamda2}, we have applied the rank-1 perturbation lemma, \cite[Lem. B.26]{Bai2010}, and \cite[Theorem 1]{Wagner2012} together with  \cite[Theorem 2]{Hoydis2013} for $\bK=\Id_M$. In \eqref{desired1}, we have denoted $ \delta_{i}^{\lambda}=\frac{1}{M}\tr\big(\bPhi_{i} \tilde{\bT}(\Id_M) \big)$.

	
	The term in $ I_{k,n} $, which includes the deviation from the average effective
	channel gain  is written as
	\begin{align}
		\mathrm{Var}\left\{{\bh}_{k}^{\H}\bff_{k,n}\right\}	&\asymp\frac{\al_{k,n-K}^{2}\EE\{|\tilde{\bh}_{k}^{\H}\bSigma_{k}\hat{\bh}_{k}|^{2}\}}{(1+\al_{k,n-K}^{2}\delta_{k})^{2}}\label{var1}\\
		&	\asymp \frac{\al_{k,n-K}^{2}\tr(\bSigma_{k}\bPhi_{k}\bSigma_{k}(\bR_{k}-\bPhi_{k}))}{(1+\al_{k,n-K}^{2}\delta_{k})^{2}}\label{var2}\\
		&	\asymp \frac{\frac{1}{M^{2}}\al_{k,n-K}^{2}\tr\big((\bR_{k}-\bPhi_{k})\tilde{\bT}(\bPhi_{k})\big)}{(1+\al_{k,n-K}^{2}\delta_{k})^{2}}\label{var3},
	\end{align}
	where in \eqref{var1}, we have applied the matrix inversion lemma, \cite[Lem. B.26]{Bai2010}, and \cite[Theorem 1]{Wagner2012}. Next step includes application of the rank-1 perturbation lemma, \cite[Lem. B.26]{Bai2010}, and \cite[Lem. 10]{Krishnan2015}. and \cite[Lem. B.26]{Bai2010} again. In   \eqref{var3}, we have used \cite[Theorem 2]{Hoydis2013}.
	
	Similar steps are followed to derive the DE of the term including the innovation error. Specifically, we have 
	\begin{align}
		\EE\left\{|	\tilde{\bee}_{k,n}^{\H}\bff_{k,n}|^{2}\right\}&\asymp\frac{\al_{k,n-K}^{2}\EE\{|\tilde{\bee}_{k,n}^{\H}\bSigma_{k}\hat{\bh}_{k}|^{2}\}}{(1+\al_{k,n-K}^{2}\delta_{k})^{2}}\label{error1}\\
		&	\asymp \frac{\al_{k,n-K}^{2}\tr(\bSigma_{k}\bPhi_{k}\bSigma_{k}\bR_{k})}{(1+\al_{k,n-K}^{2}\delta_{k})^{2}}\label{error2}\\
		&	\asymp \frac{\frac{1}{M^{2}}\al_{k,n-K}^{2}\tr\big(\bR_{k}\tilde{\bT}(\bPhi_{k})\big)}{(1+\al_{k,n-K}^{2}\delta_{k})^{2}}\label{error3}.
	\end{align}
	
	The second term, concerning the multi-user interference, is obtained as
	\begin{align}
		\EE\{|\bh^\H_{k,n}\bff_{i,n}|^{2}\}&=\al_{i,n-K}^{2}\EE\bigg\{\!\left	|\frac{\bh^\H_{k,n}\bSigma_{i}\hat{\bh}_{i}}{1+\al_{i,n-K}^{2}\hat{\bh}_{i}^{\H}\bSigma_{i}\hat{\bh}_{i}}\right|^{2}\!\bigg\}\label{int1}\\
		&=\al_{i,n-K}^{2}	\EE\bigg\{\!\frac{\bh^\H_{k,n}\bSigma_{i}\hat{\bh}_{i}\hat{\bh}_{i}^{\H}\bSigma_{i}\bh_{k,n}}{(1+\al_{i,n-K}^{2}\hat{\bh}_{i}^{\H}\bSigma_{i}\hat{\bh}_{i})^{2}}\!\bigg\}\label{int2}\\
		&\asymp\al_{i,n-K}^{2}	\EE\bigg\{\!\frac{\bh^\H_{k,n}\bSigma_{i}\bPhi_{i}\bSigma_{i}\bh_{k,n}}{(1+\al_{i,n-K}^{2}\delta_{i})^{2}}\!\bigg\},\label{int3}
	\end{align}
	where the matrix inversion lemma has been applied in \eqref{int1}. In \eqref{int3}, the mutual independence between $ \hat{\bh}_{i} $ and $\bh_{k,n}  $, the rank-1 perturbation lemma, \cite[Lem. B.26]{Bai2010}, and \cite[Theorem 1]{Wagner2012} are taken into account. However, $ \bSigma_{i} $ is not independent of $ \bh_{k,n} $. For this reason, application of \cite[Lemma~2]{Hoydis2013} gives
	\begin{align}
		\bSigma_i={\bSigma}_{ik}-\al_{k,n-K}^{2}\frac{{\bSigma}_{ik}\hat{\bh}_{k}\hat{\bh}_{k}^{\H}{\bSigma}_{ik}}{1+\al_{k,n-K}^{2}\hat{\bh}_{k}^{\H} {\bSigma}_{ik}\hat{\bh}_{k} },\label{eq:theorem2.I.51}
	\end{align}
	where the new matrix ${\bSigma}_{ik}$ is  defined as
	\begin{align}
		{\bSigma}_{ik}=&\big(\al_{k,n-K}^{2}\hat{\bH}_{k}\hat{\bH}_{k}^{\H} \!-\al_{k,n-K}^{2}\hat{\bh}_{k}\hat{\bh}_{k}^{\H}-\al_{i,n-K}^{2}\hat{\bh}_{i}\hat{\bh}_{i}^{\H}	\! \nn\\
		&+\bZ + M \al \Id_M \big)^{-1}.
	\end{align}
	
	By inserting \eqref{eq:theorem2.I.51} into \eqref{int3}, we result in 
	\begin{align}
		\EE\{|\bh^\H_{k,n}\bff_{i,n}|^{2}\}=\frac{Q_{ik}}{\big(1+\al_{i,n-K}^{2}{\delta_{i}}\big)^{2}},\label{eq:theorem2.I.6}
	\end{align}
	where $Q_{ik}$ is given by \eqref{eq:theorem2.I.mu1} at the top of the next page.
	\begin{figure*}
		\begin{align}
			\!\!\!\!\!\!\!\!\!\!\!\!Q_{ik}&\!= \!\bh^\H_{k,n}\bSigma_{ik} \bPhi_{i} \bSigma_{ik}  \bh_{k,n}\!+\!\al_{k,n-K}^{6}\frac{\big|  \bh^\H_{k,n}{\bSigma}_{ik} \hat{\bh}_{k}\big|^{2}\!\hat{\bh}^\H_{k}  {\bSigma}_{ik} \bPhi_{i}{\bSigma}_{ik}\hat{\bh}_{k}}{\big( 1+\al_{k,n-K}^{2}\hatvh^\H_{k} {\bSigma}_{ik} \hatvh_{k} \big)^{2}}\!-\!2\al_{k,n-K}^{4}\mathrm{Re}\!\left\{ \! \frac{\hatvh^\H_{k}{\bSigma}_{ik}\bh_{k,n}\bh_{k,n}^{\H}{\bSigma}_{ik}\bPhi_{i}{\bSigma}_{ik}\hatvh_{k}}{1+\al_{k,n-K}^{2}\hatvh^\H_{k} {\bSigma}_{ik} \hatvh_{k}}\!\right\}\!.
			\label{eq:theorem2.I.mu1}
		\end{align}
		\line(1,0){470}
	\end{figure*}
	In \eqref{eq:theorem2.I.6}, we have applied again the matrix inversion lemma, and in the last step, we have used the rank-1 perturbation lemma, \cite[Lem. B.26]{Bai2010}, \cite[Theorem 1]{Wagner2012}, and \cite[Theorem 2]{Hoydis2013}.  The DE of each term in~\eqref{eq:theorem2.I.mu1} is obtained as
	\begin{align}
		\bh^\H_{k,n}\bSigma_{ik} \bPhi_{i} \bSigma_{ik}  \bh_{k,n}&\asymp  \frac{1}{M^{2}}\tr \bR_{k}\tilde{\bT}(\bPhi_{i})=\zeta_{ki}\label{q1},\\
		\hat{\bh}^\H_{k}  {\bSigma}_{ik} \bPhi_{i}{\bSigma}_{ik}\hat{\bh}_{k}&\asymp  \frac{1}{M^{2}}\tr \bPhi_{k}\tilde{\bT}(\bPhi_{i})=\mu_{ki},\\
		\bh_{k,n}^{\H}{\bSigma}_{ik}\bPhi_{i}{\bSigma}_{ik}\hatvh_{k}&\asymp \al_{k,n-K} \frac{1}{M^{2}}\tr\bPhi_{k}\tilde{\bT}(\bPhi_{i}).\label{q3}
	\end{align}
	Hence, substitution of \eqref{q1}-\eqref{q3} into \eqref{eq:theorem2.I.mu1}  provides $ Q_{ik} $. 
	
	The definitions of the various parameters are given in the presentation of the theorem.  The DE SINR $\bar{\gamma}_{k,n} $ is obtained by substituting \eqref{d2}, \eqref{desired1}, \eqref{var3}, \eqref{error3}, and \eqref{eq:theorem2.I.6} into $ S_{k,n} $ and $ I_{k,n} $.
	
	\section{Proof of Proposition~\ref{Prop:optimPhase}}\label{optimPhase}
	For this proof, we are going to use the following lemma.
	\begin{lemma}\label{traceProd}
		Let 	 $ \bA \in\mathbb{C}^{M\times M} $ be independent of $ \bTheta$ and $\bR_{k}= \bH_{1} \bTheta$ $\bR_{\mathrm{RIS},k}\bTheta^{\H}\bH_{1}^{\H} $, then
		\begin{align}
			\tr\left( \!\!\bA\pdv{\bR_{k}}{\bc_{t}^{*}}\!\right) =\mu\diag\left(\bH_{1}^{\H}\bA\bH_{1} \bTheta\bR_{\mathrm{IRS},k}\right).
	\end{align}\end{lemma}
	\proof We have
	\begin{align}
		\!\tr\!\left(\! \!\bA\pdv{\bR_{k}}{\bc_{t}^{*}}\!\right) &\!=\!\pdv{\left(\diag\left(\bH_{1}^{\H}\bA\bH_{1} \bTheta\bR_{\mathrm{IRS},k}\right)\right)^{\T}\bc_{t}^{*}}{\bc_{t}^{*}}\\
		&= \mu\diag\left(\bH_{1}^{\H}\bA\bH_{1} \bTheta\bR_{\mathrm{IRS},k}\right),\label{deriv4}
	\end{align}
	where we have used the property $ \tr\left(\bA \diag(\bc_{t})\right)=\left(\diag(A)\right)^{\T}\bc_{t} $.
	\endproof
	The  gradient of $ \!\bar{\gamma}_{k,n}  $  with respect to $ \bc_{t}^{*} $ is written as
	\begin{align}
		\!\bar{\gamma}_{k,n}' =\frac{S_{k}'I_{k}-S_{k}I_{k}'}{I_{k}^{2}},\label{gam1}
	\end{align}
	where the quotient rule derivative was simply applied. We continue with the computation of the partial derivatives $ S_{k}' $ and $ I_{k}' $. Specifically, $ S_{g}' $ is obtained as
	\begin{align}
		S_{k}'= 2p_{k}\delta_{k}\delta_{k}', \label{sg1}
	\end{align}
	where 
	\begin{align}
		\delta_{k}'=\frac{1}{M}\tr\left(\bPhi_{k}'\bT+\bPhi_{k}\bT'\right).\label{sg2}
	\end{align} 
	Note that $ \bPhi_{k}' $ is the derivative of $ \bPhi_{k} $  with respect to $\bc_{t}^{*}  $ whose expression is given by
	\begin{align}
		\bPhi_{k}'=\al^{2}_{k,\zeta-n}\left(\bR_{k}' \bQ \bR_{k}+\bR_{k} \bQ' \bR_{k}+\bR_{k} \bQ \bR_{k}'\right).
	\end{align}
	The computation of the trace of this expression requires Lemma \ref{traceProd} and that the derivative of the  inverse matrix $ \bQ $ is obtained according to \cite[Eq. 40]{Petersen2012} as  $ \bQ'=-\bQ(\bQ^{-1})'\bQ$, where the derivative of $ \bQ^{-1} $ is written as
	\begin{align}
		(\bQ^{-1})'=\bR_{k}'.
	\end{align}
	Also, $ \bT' $ concerns the derivative of an inverse matrix written as $ \bT'=-\bT(\bT^{-1})'\bT$, where the derivative of $ \bT^{-1} $ is written as
	\begin{align}
		(\bT^{-1})'=\frac{1}{M}\displaystyle \sum_{i=1 }^{K}\frac{\bPhi_{i}'\left(1+\al_{i,n-K}^{2}\delta_{i}\right)-\al_{i,n-K}^{2}\bPhi_{i}\delta_{i}'}{(1+\al_{i,n-K}^{2}\delta_{i})^{2}},\label{sg3}
	\end{align}
	After inserting \eqref{sg2}-\eqref{sg3} into \eqref{sg1}, we obtain $S_{k}'  $. 
	
	The derivative of $ I_{k} $ is based on simple derivative rules and it requires the computation of the derivatives of $ \delta_{k} $, $ \tilde{\delta}_{k} $,	$ \delta_{k}^{\mathrm{e}} $, $ Q_{ik} $, and $ 	\bar{\lambda} $ as shown in \eqref{intDeriv}. Regarding $ \delta_{k}' $, it is given by \eqref{sg2}, while $ \tilde{\delta}_{k} $ and	$ \delta_{k}^{\mathrm{e}}$ are obtained similarly as
	\begin{align}
		\!\!\!\! \tilde{\delta}_{k}'&\!=\!\frac{1}{M}\tr\!\big(\!(\bR_{k}'-\bPhi_{k}')\tilde{\bT}(\bPhi_{k})\!+\!(\bR_{k}-\bPhi_{k})\tilde{\bT}'(\bPhi_{k})\!\big),\label{sg4}\\
		( \delta_{k}^{\mathrm{e}})'&\!=\!\frac{1}{M}\tr\big(\bR_{k}'\tilde{\bT}(\bPhi_{k})+\bR_{k}\tilde{\bT}'(\bPhi_{k})\big),\label{sg5}
	\end{align}
	where 
	\begin{align}
		\tilde{\bT}'(\bPhi_{k})&=\bT'\bPhi_{k}\bT+\bT\bPhi_{k}'\bT+\bT\bPhi_{k}\bT'\nn\\
		&+\frac{1}{M}\sum_{i=1}^{K}\frac{\tilde{\delta}_{i}'\bT\bPhi_{i}\bT+\tilde{\delta}_{i}(\bT'\bPhi_{k}\bT+\bT\bPhi_{k}'\bT+\bT\bPhi_{k}\bT')}{(1+\al_{i,n-K}^{2}\delta_{i})^{3}}\nn\\
		&-\frac{1}{M}\sum_{i=1}^{K}\frac{2\al_{i,n-K}^{2}\tilde{\delta}_{i}\bT\bPhi_{i}\bT \delta_{i}'}{(1+\al_{i,n-K}^{2}\delta_{i})^{}}.
	\end{align}
	
	The derivative of the normalization parameter is given by
	\begin{align}
		\bar{\lambda}'&=\frac{K \big( \frac{1}{M} \tr \left(\frac{\Zm}{M} +\al \Id_M\right) \tilde{\bT}'(\Id_{M})-\frac{1}{M}\tr\Tm'\big)}{ \big(\frac{1}{M}\tr\Tm - \frac{1}{M} \tr \left(\frac{\Zm}{M} +\al \Id_M\right) \tilde{\bT}(\Id_{M})\big)^{2}}.\label{lambdaderiv}
	\end{align}
	
	Moreover, the last required derivative, concerning $ Q_{ik} $, is obtained as
	\begin{align}
		Q_{ik}'&= \zeta_{ki}'\!+\!\frac{\delta_{k}\delta_{k}'\mu_{ki}+|\delta_{k}|^{2}\mu_{ki}}{\big(1\!+\!\al_{k,n-K}^{2}\delta_{k}\big)^{\!2}}\!+\!\frac{2\al_{k,n-K}^{2}|\delta_{k}|^{2}\mu_{ki}\delta_{k}'}{\big(1\!+\!\al_{k,n-K}^{2}\delta_{k}\big)^{\!3}}\nn\\
		&-\!2\mathrm{Re}\left\{\!\!
		\frac{ (\delta_{k}^{*})'\mu_{ki}+\delta_{k}^{*}\mu_{ki}'}{1\!+\!\al_{k,n-K}^{2}\delta_{k}}+\!\!\!
		\frac{ \al_{k,n-K}^{2}\delta_{k}^{*}\mu_{ki}\delta_{k}'}{\big(1\!+\!\al_{k,n-K}^{2}\delta_{k}\big)^{\!2}}\right\},\label{qik}
	\end{align}
	where 
	\begin{align}
		\zeta_{ki}'&=\frac{1}{M^{2}}\tr\big(\bR_{k}'\tilde{\bT}(\bPhi_{i})+\bR_{k}\tilde{\bT}'(\bPhi_{i})\big), \\
		\mu_{ki}'&=\frac{1}{M^{2}}\tr\big(\bPhi_{k}'\tilde{\bT}(\bPhi_{i})+\bPhi_{k}\tilde{\bT}'(\bPhi_{i})\big).
	\end{align}
	Having obtained \eqref{sg4}, \eqref{sg5}, \eqref{lambdaderiv}, and \eqref{qik}, the derivative of $ I_{k} $ is derived and the proof is concluded.
\end{appendices}
\bibliographystyle{IEEEtran}

\bibliography{mybib}

\end{document}